\documentclass[letterpaper, 10 pt, conference]{ieeeconf}  

\IEEEoverridecommandlockouts                              
\usepackage{amsmath} 
\usepackage{amssymb}  

\title{\LARGE \bf
On the stability of event-based control with neuronal dynamics
}

\author{Luke Eilers, Jonas Stapmanns, Catarina Dias, and Jean-Pascal Pfister
\thanks{This work has been supported by the Swiss National Science Foundation grant entitled ``Why spikes?'' (310030\_212247).}
\thanks{Luke Eilers and Catarina Dias are with the Graduate School for Cellular and Biomedical Sciences, University of Bern, 3012 Bern, Switzerland.}
\thanks{All authors are with the Department of Physiology, University of Bern, 3012 Bern, Switzerland.}%
\thanks{Correspondence: \{{\tt\scriptsize luke.eilers},{\tt\scriptsize jeanpascal.pfister}\}{\tt\scriptsize @unibe.ch}}
}

\usepackage{mathtools}
\usepackage{dsfont}
\usepackage[dvipsnames]{xcolor} 
\usepackage{accents}
\usepackage{aligned-overset}

\usepackage{cite}

\makeatletter
\def\@begintheorem#1#2{\par\vspace{0.4\baselineskip}\noindent\textbf{#1~#2 } \itshape}
\def\@opargbegintheorem#1#2#3{\par\vspace{0.4\baselineskip}\noindent\textbf{#1~#2~(#3):} \itshape}
\def\@endtheorem{\par\vspace{0.0\baselineskip}} 
\makeatother

\newtheorem{thm}{\bf Theorem} 

\newtheorem{lem}[thm]{\bf Lemma}
\newtheorem{cor}[thm]{\bf Corollary}
\newtheorem{rem}[thm]{\bf Remark}

\newtheorem{definition}[thm]{\bf Definition}

\newcommand{\R}{\mathbb{R}}

\newcommand{\del}{\textnormal{d}}

\newcommand{\xc}{x_c}
\newcommand{\ui}[1]{\mathcal{I}^{#1}}
\newcommand{\xcd}[1]{#1 x_c}

\newcommand{\ubar}[1]{\underaccent{\bar}{#1}}

\begin{document}

\maketitle
\thispagestyle{empty}
\pagestyle{empty}

\begin{abstract}

Event-based control, unlike analogue control, poses significant analytical challenges due to its hybrid dynamics. This work investigates the stability and inter-event time properties of a control-affine system under event-based impulsive control. The controller consists of multiple neuronal units with leaky integrate-and-fire dynamics acting on a time-invariant, multivariable plant in closed loop. Both the plant state and the neuronal units exhibit discontinuities that cancel if combined linearly, enabling a direct correspondence between the event-based impulsive controller and a corresponding analogue controller. Leveraging this observation, we prove global practical stability of the event-based impulsive control system. In the general nonlinear case, we show that the event-based impulsive controller ensures global practical asymptotic stability if the analogue system is input-to-state stable (ISS) with respect to specific disturbances. In the linear case, we further show global practical exponential stability if the analogue system is stable. We illustrate our results with numerical simulations. The findings reveal a fundamental link between analogue and event-based impulsive control, providing new insights for the design of neuromorphic controllers.
\end{abstract}

\section{Introduction}

Both neuromorphic engineering and event-based control draw inspiration from biological systems~\cite{mead_neuromorphic_1990, mead2012analog, deweerth1991simple,aaarzen1999simple}. Unlike periodic sampling, event-based sampling triggers updates only when necessary, reducing sensing, communication, and computation effort~\cite{astrom2002comparison,heemels2012anintroduction, antunes2016consistent}, which is particularly beneficial in distributed and networked systems~\cite{mahmoud2014networked}. Efficiency is likewise critical in the brain, where neurons organised in networks communicate via electrical pulses, so-called spikes, which makes brains inherently event-based and impulsive~\cite{rieke1999spikes}. Moreover, the brain can be viewed as a closed-loop control system~\cite{ahissar2016perception, cisek1999beyond, moore2024theneuron}.
Despite this strong biological motivation, theoretical guarantees for neuromorphic and event-based control systems remain difficult to obtain due to their hybrid dynamics,
even though significant progress has been made in the last two decades~\cite{tabuada2007event, aranda2020bibliometric, heemels2012anintroduction}.

Several approaches have been developed to analyse event-based control systems. A common strategy is the hybrid systems framework using Lyapunov functions~\cite{goebel2012hybrid, haddad_impulsive_2014}. If the Lyapunov function decreases during both continuous and discontinuous dynamics, stability and gain properties follow \cite{donkers2012output, heemels2013periodic, postoyan2011lyapunov}. Alternative approaches model the system as piecewise-linear~\cite{heemels2008analysis,heemels2013periodic, heemels2012anintroduction} or perturbed-linear (PL) systems~\cite{heemels2013periodic}.

In contrast, event-based impulsive control is less explored but particularly interesting, since impulsive control inherently involves discrete events: each control action is an impulse applied at an event time, where both the timing and amplitude determine the closed-loop behaviour \cite{yang2002impulsive}. The seminal work \cite{astrom2002comparison}, later extended in \cite{meng2012optimal}, quantified the advantage of event-based over periodic sampling in a stochastic setting by considering event-based impulsive control. From a biological perspective, event-based impulsive control also provides a natural abstraction for spike-based controllers. 
The stability analysis of such systems is more challenging, since possible increases of the Lyapunov function between impulses have to be compensated by decreases due to the impulses (see \cite{yang2002impulsive, petri2024analysis, chai2017analysis, li2020lyapunov} and Chapter 3.2 of \cite{goebel2012hybrid}). This trade-off is often studied by discretising the dynamics at the event times~\cite{petri2024analysis}.

We study an event-based impulsive controller with neuronal dynamics. Similar setups have been investigated both in neuromorphic engineering \cite{deweerth1991simple,petri2024analysis} and computational neuroscience \cite{agliati2025spiking, slijkhuis2023closed, huang2018dynamical}, where it has been shown that neuronal dynamics can arise from a normative impulsive control law with a finite set of admissible impulses \cite{agliati2025spiking}. However, the stability properties of such systems remain only partially understood, limiting the applicability of neuromorphic controllers in reliable control design. We tackle this issue by proving practical stability and the existence of minimum inter-event times for both linear and nonlinear dynamics, thereby extending the results of \cite{petri2024analysis}. 

Our key observation is that the plant state and neuronal variables can be combined linearly such that their discontinuities cancel, similar to the approach in \cite{thalmeier2016learning}. The dynamics of the resulting continuous variable can be interpreted as a perturbed version of a corresponding analogue control system, similar to the approach of \cite{petri2025spiking,petri2025emulation}, where the event-based impulsive control system is viewed as an analogue control system with hybrid perturbations. A similar viewpoint is also taken in the PL systems approach~\cite{heemels2013periodic}, where a discretised event-based system is viewed as being perturbed compared to the control system with periodic sampling. 
This perspective allows us to analyse the system in continuous time using standard arguments without discretisation at event times. Importantly, we show that the stability of the event-based impulsive control system can be directly inferred from the corresponding analogue control system, thereby providing a novel method for the stability analysis and revealing a fundamental link between the two control paradigms.

In many works, event triggering is based on Lyapunov functions \cite{tabuada2007event, dimarogonas2012distributed, li2020lyapunov, wang2011event}. In contrast, our triggering mechanism is biologically inspired and effectively implements quantisation with hysteresis. This makes our approach similar to \cite{kofman2006level}, where quantisation occurs at both the sensor and actuator. However, we consider impulsive control rather than zero-order hold, and there is only one quantisation step.

The paper is organised as follows. Section \ref{sec:definition} provides the necessary notations and definitions, and formalises the event-based impulsive control system. In Section \ref{sec:theory}, we provide stability guarantees (\ref{subsec:stability}), lower bounds for the inter-event times (\ref{subsec:interevent}), and an extension to connected neuronal units (\ref{subsec:extension}). Section \ref{sec:numerics} illustrates the theoretical results with numerical simulations for linear systems. We conclude by discussing possible extensions to our work in Section \ref{sec:conclusion}.

\section{Preliminaries}
\label{sec:definition}

\subsection{Notation and definitions}

We denote by $\R_{\geq 0}$ the non-negative real numbers. We denote by $\mathcal{I} \coloneqq [0,1]{}$ the unit interval. We denote by $\lVert x \rVert$ the Euclidean norm of a vector $x\in \R^n$ and by $\lVert A\rVert$ the spectral norm of a matrix $A \in \R^{n\times m}$. We denote by $A^\intercal \in \R^{m\times n}$ the transpose of the matrix $A$. We denote by $I$ the identity matrix.
For $a_1,\dots,a_n \in \R$, $D =\mathrm{diag}(a_1,\dots,a_n) \in \R^{n\times n}$ denotes the diagonal matrix with entries $D_{ii} = a_i$ for all $i = 1,\dots,n$. We write $A \succ 0$ ($A \succeq 0$) if the matrix $A$ is symmetric positive (semi-)definite. We denote by $\lambda_\mathrm{min}(A)$ and $\lambda_\mathrm{max}(A)$ the smallest and the largest eigenvalue of $A$, respectively. We denote by $\kappa(A) \coloneqq \lambda_\mathrm{max}(A) /\lambda_\mathrm{min}(A)$ the condition number of the matrix $A$. We say that a matrix $A$ is Hurwitz if $\mathrm{Re(\lambda_i)<0}$ for all eigenvalues $\lambda_i$ of $A$. We denote by $\delta(t)$ the Dirac delta function. We denote by $[x]_+ \coloneqq\max(x,0)$ the positive part of $x \in \R$.

We say that a continuous function $\alpha \colon [0,a) \to [0,\infty)$ is of class $\mathcal{K}$ if it is strictly increasing and $\alpha(0)=0$. We say that a function $\alpha$ of class $\mathcal{K}$ is of class $\mathcal{K}_\infty$ if $a = \infty$ and $\lim_{r \to \infty} \alpha(r) = \infty$. We say that a function $\beta \colon [0,\infty)\times [0,\infty) \to [0, \infty)$ is of class $\mathcal{KL}$ if for any $t \geq 0$ the mapping $r \mapsto \beta(r,t)$ is of class $\mathcal{K}$, and for any $r \geq 0$ the mapping $t \mapsto \beta(r,t)$ is decreasing and $\lim_{t \to \infty} \beta(r,t) = 0$.

\begin{definition}[Global practical stability]
Consider the dynamical system $\dot x(t) = h(x(t))$, $x(t)\in \R^n$, with $h(0)=0$. We say that the system is
    \begin{enumerate}
        \item globally practically asymptotically stable if there exists a  function $\beta$ of class $\mathcal{KL}$ and constant $C\geq0$ such that
        \begin{equation}
            \lVert x(t) \rVert \leq \beta(\lVert x(0) \rVert,t)+C, \notag
        \end{equation}
        \item globally practically exponentially stable if there exist constants $\alpha < 0$ and $D,C\geq0$ such that
        \begin{equation}
            \lVert x(t) \rVert \leq D\,\lVert x(0)\rVert \cdot e^{\alpha t} + C. \notag
        \end{equation}
    \end{enumerate}
    In both cases, it holds that $\limsup_{t\to\infty} \lVert x(t)\rVert \leq C$ and we call $C$ the ultimate bound. If $C = 0$, the term ``practically'' is omitted.
\end{definition}

\subsection{Control loop and state dynamics}

We consider a multivariable control-affine system,
\begin{equation}
    \dot x(t) = f(x(t)) + u(t), \quad x(0) = x_0, \notag
\end{equation}
where $x(t) \in \R^K$, $K\geq 1$, is the state, $u(t) \in \R^K$ the control input, and $f : \R^K \to \R^K$ is a linear or nonlinear drift function. The objective of the control input is to drive the state towards a target state $x^*(t)$. We denote the error between state and target state by $e(t) = x(t) - x^*(t)$. In the case of analogue control, we assume that the control input is a function of the error, i.e., $u(t) = k(e(t))$. For instance, a proportional controller is given by $k(e) = -\alpha \,e$, $\alpha > 0$.

Here, we consider an event-based and impulsive control mechanism using $N$ neuronal units, which generate events upon threshold crossing of their respective neuronal variables. The event times determine when impulses are applied to the state.  If $t_j^i$ denotes the $j$-th event time of unit $i$, we can write the so-called spike trains as $s_i(t) = \sum_j \delta(t-t_j^i)$, $i \in \{1,\dots,N\}$. An input matrix $B \in \R^{K \times N}$ then maps the vector of spike trains $s(t)$ to the state space. This defines the impulsive control input, $u(t) = B s(t)$. The dynamics of the state are hence given by
\begin{equation}
    \dot x(t) = f(x(t)) + Bs(t), \quad x(0) = x_0.  \label{eq:dyn_sys_ebic:1}
\end{equation}

In the absence of events, the system evolves according to $\dot x(t) = f(x(t))$. Assuming right-continuity of all variables, at event time $t=t_j^i$ it holds that $x(t) = \lim_{s \uparrow t} x(s) + B_i$, where $B_i$ denotes the $i$-th column of $B$.
This description is used by the hybrid systems formalism \cite{goebel2012hybrid, haddad_impulsive_2014}.
Here, we instead use the notation \eqref{eq:dyn_sys_ebic:1} to propose a novel method.


We denote the neuronal variables by $z_i(t)$, and the respective thresholds by $\theta_i>0$, $i \in \{1,\dots,N\}$. Then, the event-triggering by threshold crossing can be written as
\begin{equation}
    s_i(t) = \delta_{z_i(t) = \theta_i} \quad \text{for } i \in \{1,\dots,N\} .\label{eq:dyn_sys_ebic:2}
\end{equation}
The dynamics of $z_i$ are inspired by the leaky integrate-and-fire neuron model \cite{gerstner2002spiking}, and are given by
\begin{equation}
    \dot z_i(t) = - \lambda_i z_i(t) + g_i(t) - \theta_i s_i(t), \quad z_i(0)=0,
    \label{eq:LIF_dynamics}
\end{equation}
where $\lambda_i \geq 0$ is the leakage constant and $g_i(t)$ is the input being integrated. Here, we assume that $g_i(t)$ is a function of the error $e(t)$ and non-negative, $g_i(t) = g_i(e(t)) \geq 0$. Upon reaching the threshold, each neuronal variable is reset to zero due to the term $-\theta s_i(t)$, such that $z_i(t)\in [0,\theta_i]$. Hence, the neuronal variables can be thought of as membrane potentials. The closed loop between plant and controller is illustrated in Fig. \ref{fig:1a_control_loop}, and example trajectories are presented in Fig. \ref{fig:2_1dim} a).

We can write the neuronal variables' dynamics in vector notation. Let $\Theta = \mathrm{diag}(\theta_1,\dots,\theta_N)$ be the reset matrix and $\Lambda = \mathrm{diag}(\lambda_1,\dots,\lambda_N)$ the leakage matrix. Moreover, let $g \colon \R^K \to \R_{\geq 0}^N$ be the input function and let $z(t) \in \R^N$ be the vector of neuronal variables. Then, we obtain
\begin{equation}
    \dot z(t) = -\Lambda z(t) + g(e(t)) - \Theta s(t), \quad z(0) = 0. \label{eq:dyn_sys_ebic:3}
\end{equation}

The key observation for our method is that both the state and the neuronal variables have the same source of discontinuities, $s(t)$. If we consider the auxiliary variable $\xc(t) \coloneqq x(t) + B\Theta^{-1}z(t)$, the discontinuities cancel out:
\begin{alignat}{2}
    \xcd{\dot}(t) &= &&\dot x(t) + B\Theta^{-1} \dot z(t)  \notag \\
    &= &&f(x(t)) + Bs(t) - B\Theta^{-1}\Lambda z(t)  \notag \\
    &&&+ B\Theta^{-1}g(e(t)) -B\Theta^{-1}\Theta s(t)  \notag \\
    &=&& f(x(t)) + B\Theta^{-1}g(e(t)) - B\Theta^{-1}\Lambda z(t) . \label{eq:auxil_var}
\end{alignat}
We write ``c'' since $x_c$ is the continuous version of $x$. Fig.~\ref{fig:2_1dim} and Fig. \ref{fig:3_2dim} show trajectories of $x(t)$ and $x_c(t)$. 

To interpret the auxiliary variable $\xcd{}(t)$, imagine that at any time $t$, the $i$-th neuronal unit could be ``asked'' to generate an event and emit a hypothetical impulse rescaled by how close the neuronal variable is to its threshold, i.e., $B_i z_i(t)/\theta_i$. 
Applying the hypothetical impulses of all neuronal units to the state amounts to a jump from $x(t)$ to $\xcd{}(t)$.

In the following section, we will analyse the practical stability of $x(t)$, which can be inferred from the practical stability of $\xc(t)$ by noticing that
\begin{align}
    \lVert \xc(t) - x(t) \rVert &= \lVert B\Theta^{-1} z(t)\rVert  \notag \\
    &\leq \sup_{s \in \ui{N}} \lVert Bs \rVert \eqqcolon \lVert B \rVert_{\ui{N}}, \label{eq:bound_y-x}
\end{align}
where we used that $z_i(t)/\theta_i \in \mathcal{I}$ for all $i \in \{1,\dots,N\}$.

\begin{figure}[htbp]
\centering
\vspace{2.5mm}
{\includegraphics[width=0.9\columnwidth]{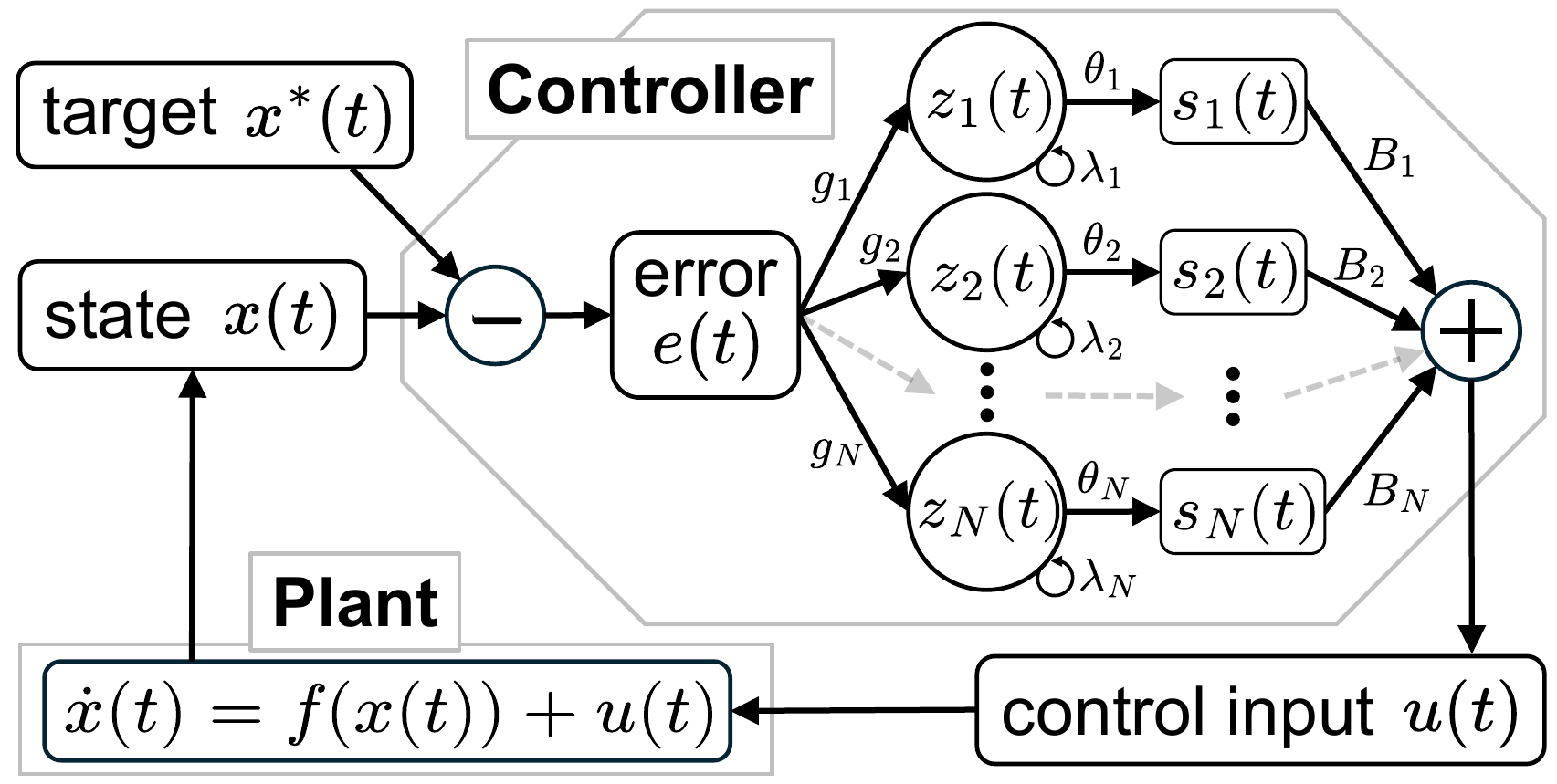}}
\caption{
The control loop between the controller, which receives the error and generates a control input, and the plant. 
}
\label{fig:1a_control_loop}
\end{figure}

\section{Theoretical results}
\label{sec:theory}

\subsection{Stability}
\label{subsec:stability}

The event-based impulsive control system evolves uncontrolled between events, so a Lyapunov function typically decreases only at the impulses. Such systems can be analysed within the hybrid-systems framework by discretising them at the event times \cite{petri2024analysis}. This approach works well when the evolution between two events, $t_k$ and $t_{k+1}$, depends only on the initial state $x(t_k)$, as in the setup of \cite{petri2024analysis} with two neuronal units. For an arbitrary number of units, however, the analysis quickly becomes intractable unless the units are sequentially active. Instead, we will consider a Lyapunov function of $x_c$, which eliminates the need for discretisation.


In the following analysis, we will consider only a constant target state $x^*(t) = 0$ for the sake of simplicity. The error is then given by $e(t) = x(t)$. 

We will first apply our method to linear dynamics with scalar states, then to linear dynamics with multivariable states, and finally to general nonlinear dynamics. For the first case, we assume that $f$ is linear, the state is scalar, $K=1$, and that we have two neuronal units, $N=2$. 
This is already a slight generalisation of the setup in \cite{petri2024analysis}, where a specific choice of the input function $g$ was analysed.

\begin{thm}[Linear dynamics, scalar state]
\label{thm:stability_linear_scalar}
    Let $x(t) \in \R$ be the state of an event-based impulsive control system as in \eqref{eq:dyn_sys_ebic:1}--\eqref{eq:dyn_sys_ebic:3} with input matrix $B \in \R^{1\times 2}$, input function $g \colon \R \to \R_{\geq 0}^2$, reset matrix $\Theta=\mathrm{diag}(\theta,\theta)$ with $\theta>0$, leakage matrix $\Lambda = \mathrm{diag}(\lambda,\lambda)$ with $\lambda \geq 0$, and drift function $f(x) = ax$ with $a \in \R$. Let $k(x) \coloneqq B\Theta^{-1}g(x)$ be linear, i.e., $k(x)=-bx$ for some $b \in \R$. If $b>a$, the system is globally practically exponentially stable,
    \begin{equation}
        | x(t) | \leq | x(0)| \cdot e^{(a-b)\cdot t/2} + \lVert B \rVert_{\ui{2}}\left(1+\frac{|a-b+\lambda|}{|a-b|} \right) .
        \label{eq:thm_stab_lin_0}
    \end{equation}
\end{thm}

\begin{proof}
    With $\xc(t) \coloneqq x(t) + B\Theta^{-1}z(t)$, we see that
    \begin{align}
        \xcd{\dot}(t) &= \dot x(t) + B\Theta^{-1}\dot z(t)  \notag \\
        \overset{\eqref{eq:auxil_var}}&{=} ax(t) -bx(t) -\lambda Bz(t)/\theta  \notag \\
        &= (a-b)\xc(t) - (a-b+\lambda)Bz(t)/\theta . \label{eq:thm_stab_lin_1}
    \end{align}
    Let $d\coloneqq a-b+\lambda$. We consider the Lyapunov function candidate $V(\xc) = \xc^2$, and obtain
    \begin{align}
        \dot V(\xc(t)) &= 2\xc(t)\xcd{\dot}(t) \notag \\
        \overset{\eqref{eq:thm_stab_lin_1}}&{=} 2(a-b)\xc(t)^2 - 2d \xc(t) B (z(t)/\theta)  \label{eq:thm_stab_lin_2}\\
        \overset{\eqref{eq:bound_y-x}}&{\leq} 2(a-b)\xc(t)^2 + 2|d| |\xc(t)| \lVert B \rVert_{\ui{2}} \notag  \\
        &\leq 2(a-b)\xc(t)^2 + |d|\big(\eta \xc(t)^2 + \lVert B \rVert_{\ui{2}}^2/\eta \big)  \notag  \\
        &\leq \big(2(a-b)+|d|\eta \big)\xc(t)^2 + |d|\lVert B \rVert_{\ui{2}}^2/\eta , \notag 
    \end{align}
    for any $\eta>0$. Here, we used Young's inequality for products, $2ab\leq \eta \, a^2+b^2/\eta$ for any $a,b\geq 0$. This implies with $\alpha \coloneqq 2(a-b)+|d|\,\eta$ and $C\coloneqq |d| \cdot\lVert B \rVert_{\ui{2}}^2/\eta$ that
    \begin{equation}
        \dot V(\xc(t)) \leq \alpha V(\xc(t)) + C, \notag 
    \end{equation}
    which yields for $U(t) = V(\xc(t)) + C/\alpha$ that
    \begin{equation}
        \dot U(t) = \dot V(\xc(t)) \leq \alpha V(\xc(t)) + C = \alpha U(t). \notag 
    \end{equation}
    The comparison lemma (see Lemma \ref{lem:comparison} in the Appendix) shows that
    \begin{align}
        \xc(t)^2 
        &= V(t) 
        = U(t) -C/\alpha 
        \leq U(0)e^{\alpha t} - C/\alpha \notag \\
        &=(\xc(0)^2 + C/\alpha)e^{\alpha t} - C/\alpha \notag \\
        &= \xc(0)^2 e^{\alpha t} - C/\alpha(1-e^{\alpha t}). \label{eq:thm_stab_lin_3}
    \end{align}
    If $\alpha<0$ and $t\to\infty$, \eqref{eq:thm_stab_lin_3} equals $-C/\alpha$. We minimise $-C/\alpha$ by choosing $\eta = \frac{b-a}{|d|}>0$, which yields $\alpha = a-b<0$ and $-C/\alpha = \lVert B\rVert^2_{\ui{2}} |d|^2/(a-b)^2$. Finally, we obtain
    \begin{align}
        | x(t) | &\leq | \xc(t) | + \lVert B \rVert_{\ui{2}} \notag \\
        &\leq | x(0)| \cdot e^{(a-b)\cdot t/2} + \lVert B \rVert_{\ui{2}}\left(1+\frac{|d|}{|a-b|} \right). \notag 
    \end{align}
\end{proof}

If $\lambda= b-a$, the term $1+\frac{|a-b+\lambda|}{|a-b|}$ is minimal, and we obtain $\lVert B \rVert_{\ui{2}}$ as the ultimate bound. Conversely, the ultimate bound tends to infinity as $\lambda \to \infty$. In the case without leakage, $\lambda=0$, the ultimate bound is given by $2\lVert B \rVert_{\ui{2}}$. 

How can the analogue control input $k(x) = B\Theta^{-1}g(x)$ be interpreted? If we consider \eqref{eq:LIF_dynamics} with a constant input $g_i(x(t))$ and without leakage, the event rate of unit $i$ is given by  $r_i(t) = g_i(x(t))/\theta_i$, which yields $k(x) = Br(x)$. When many events are generated, $g(x(t)) \gg 0$, the leakage is negligible and the rate approximates the impulsive output well. Hence, it holds that $Bs(t) \approx k(x)$ when integrated. The regime where many events are generated is precisely the one that determines the stability, which is why the stability of the corresponding analogue system implies the global practical stability of the event-based impulsive control system.


The ultimate bound in \eqref{eq:thm_stab_lin_0} reflects that the impulsive control is better approximated by a rate when the impulses are small. If the uncontrolled dynamics are unstable, the ultimate bound can in general never be zero and exponential stability cannot be shown, as already observed in \cite{petri2024analysis}. The ultimate bound increases as the leakage increases, which is to be expected since the units will generate fewer events. On the other hand, whether or not the system is practically stable does not depend on the choice of $\lambda$, since the leakage becomes negligible when the time between events is small.

Next, we provide simple conditions that guarantee an ultimate bound of $2\lVert B \rVert_{\ui{2}}$ and $\lVert B \rVert_{\ui{2}}$, respectively:

\begin{cor}[Simplified bound for Theorem \ref{thm:stability_linear_scalar}]
\label{cor:stability_improved}
    If the assumptions of Theorem \ref{thm:stability_linear_scalar} hold and $b \geq a + \lambda$, the event-based impulsive control system is globally practically exponentially stable,
    \begin{equation}
        | x(t) | \leq | x(0)|\cdot e^{(a-b)\cdot t/2} + 2\lVert B \rVert_{\ui{2}} .
        \label{eq:cor_lin_1}
    \end{equation}
\end{cor}
\begin{proof}
    If $b \geq a + \lambda $, it holds that $|a - b + \lambda| = b-a-\lambda$, and this yields
    \begin{align}
        1+ \frac{|a-b+\lambda|}{|a-b|} = 1+ \frac{b-a-\lambda}{b-a} = 2-\frac{\lambda}{b-a} \leq 2. \notag 
    \end{align}
\end{proof}

\begin{thm}[Tighter bound for Theorem \ref{thm:stability_linear_scalar}]
\label{thm:stability_improved_2}
    Let the assumptions of Theorem \ref{thm:stability_linear_scalar} hold, and let $g(x)=(g_1(x),g_2(x))$. If $b \geq a + \lambda$, and if $g_1(x)=0$ if $x\leq0$, and $g_2(x)=0$ if $x\geq0$, then the event-based impulsive control system is globally practically exponentially stable,
    \begin{equation}
        | x(t) | \leq | x(0) | \cdot e^{(a-b)\cdot t} + \lVert B \rVert_{\ui{2}}.
        \label{eq:cor_lin_2}
    \end{equation}
\end{thm}
\begin{proof}
    The full proof is given in the Appendix (Theorem \ref{thm:stability_improved_appendix}). Here, we sketch the proof for the case where the uncontrolled dynamics are divergent, i.e. $a>0$: Then, the impulses bring the system closer to zero, and $B(z(t)/\theta)$ has the opposite sign of $x(t)$ by our assumptions on the input function $g$.
    Assuming without loss of generality that $x(0)>0$, we denote by $t_1$ the first event which switches the sign of $x$. We divide the proof into two steps. First, we consider $t \in [0,t_1)$, where $x(t)>0$.
    Using \eqref{eq:thm_stab_lin_1}, it follows $x_c(t) \leq x(0) \cdot e^{(a-b)\cdot t}$ since $b \geq a+\lambda$ and $B(z(t)/\theta) \leq 0$. Using the comparison lemma and \eqref{eq:bound_y-x}, this yields \eqref{eq:cor_lin_2}. In the second step, once the sign of $x(t)$ changes, it can be shown that the sign continues to change with every event if the amplitude of the impulses is the same for both neuronal units. This implies $|x(t)| \leq \lVert B \rVert_{\ui{2}}$, so that also for $t\in[t_1,\infty)$ the bound \eqref{eq:cor_lin_2} is satisfied. The other cases, i.e. convergent dynamics, $a<0$, impulses away from the target, and unequal amplitudes are covered in the full proof.
\end{proof}

\begin{rem}[Comparison to existing results]
\label{rem:literature}
    We compare our results to \cite{petri2024analysis}, and consider $B = [-\alpha,\alpha]$ for some $\alpha>0$, $g(x) = \big[[x]_+, [-x]_+ \big]^\intercal$, such that $b = \alpha/\theta$, and $\lVert B \rVert_{\ui{2}} = \alpha$. In \cite{petri2024analysis}, semi-global practical exponential stability was shown with ultimate bound $2\alpha$ under the assumption that $\theta < \rho \alpha/(\lambda +a)$ for some $\rho \in (0,1)$, and with a constraint on $|x(0)|$ dependent on $\rho$. We showed global practical exponential stability for $b>a$, which is equivalent to $\theta < \alpha/a$.
    If $\theta < \alpha/(\lambda +a)$, 
    Corollary~\ref{cor:stability_improved} guarantees an ultimate bound of $2\alpha$, and Theorem~\ref{thm:stability_improved_2} guarantees a tighter ultimate bound of $\alpha$ due to the choice of the input function $g$.
    Hence, we extend the results of \cite{petri2024analysis} by showing practical stability under weaker conditions with an improved ultimate bound.
    Finally, \cite{petri2024analysis} show the speed of convergence with respect to the event times indexed by $j$, i.e., $| x(t) | \leq \gamma^j | x(0) | + 2\alpha$, for some $\gamma \in (0,1)$ and $t\in[t_j,t_{j+1}]$, which has to be compared to the minimal and maximal inter-event times to analyse the speed of convergence with respect to continuous time. Here, we directly consider the latter.
\end{rem}

We will now consider the generalisation of Theorem \ref{thm:stability_linear_scalar} to multivariable states and possibly different thresholds and leakage constants for all $N$ neuronal units.

\begin{thm}[Linear dynamics, multivariable state]
\label{thm:stability_multi}
    Let $x(t) \in \R^K$ be the state of an event-based impulsive control system as in \eqref{eq:dyn_sys_ebic:1}--\eqref{eq:dyn_sys_ebic:3} with input matrix $B \in \R^{K\times N}$, input function $g \colon \R^K \to \R_{\geq 0}^N$, reset matrix $\Theta \in \R^{N \times N}$, leakage matrix $\Lambda \in \R^{N \times N}$, and drift function $f(x) = Ax$ for $A \in \R^{K \times K}$. Let $k(x) \coloneqq B\Theta^{-1}g(x)$ be linear, i.e., $k(x)=-Kx$ for some $K \in \R^{K \times K}$. If $A-K$ is Hurwitz, then there exists a matrix $P \succ 0$ such that the system is globally practically exponentially stable,
    \begin{align}
        \Vert x(t) \rVert \leq & \lVert x(0) \rVert \sqrt{\kappa(P)} \cdot e^{-\lambda^{-1}_\mathrm{max}(P)\cdot t/4} + \lVert B \rVert_{\ui{N}}  \notag \\
        &+ 2 \sqrt{\kappa(P)}\left\lVert P\big((A-K)B+B\Lambda \big) \right\rVert_{\ui{N}} ,
        \label{eq:thm_stability_multi_1}
    \end{align}
    where $\lambda_\mathrm{max}(P)$ and $\kappa(P)$ are the smallest eigenvalue and the condition number of $P$, respectively.
\end{thm}
\begin{proof}
    We again have $\xc(t) = x(t)+B\Theta^{-1}z(t)$, and
    \begin{align}
        \xcd{\dot}(t) 
        \overset{\eqref{eq:auxil_var}}&{=} (A-K) \big( \xc(t)-B\Theta^{-1}z(t) \big) - B\Theta^{-1}\Lambda z(t) \notag \\
        &= (A-K)\xc(t) - ((A-K)B + B\Lambda)\Theta^{-1}z(t), \label{eq:thm_stab_mult_1}
    \end{align}
    where we used that $\Theta^{-1}\Lambda = \Lambda \Theta^{-1}$. Let $E\coloneqq -((A-K)B + B\Lambda)$. Since $A-K$ is Hurwitz, for any matrix $Q \succ 0$ there exists a unique matrix $P \succ 0$ such that the following Lyapunov equation is satisfied (see Theorem 4.6 of \cite{khalil2002nonlinear}):
    \begin{equation}
        P(A-K) + (A-K)^\intercal P = -Q. \label{eq:thm:stability_multi_2}
    \end{equation}
    For simplicity, we choose $Q = I$, the identity matrix. The final decay rate can be improved by optimising over $Q$. We consider the Lyapunov function candidate $V(\xc) = \xc^\intercal P \xc$:
    \begin{alignat}{2}
        &&&\dot V(\xc(t)) = \xc(t)^\intercal P \xcd{\dot}(t) + \xcd{\dot}(t)^\intercal P \xc(t) \notag \\
        &= &&\xc(t)^\intercal (P(A-K)+(A-K)^\intercal P)\xc(t) \notag \\
        &&&+ 2\xc(t)^\intercal PE\Theta^{-1}z(t)\notag  \\
        \overset{\eqref{eq:thm:stability_multi_2}}&{\leq} &&-\xc(t)^\intercal \xc(t) + 2|\xc(t)^\intercal P E \Theta^{-1}z(t)| \notag \\
        &\leq &&-\lVert \xc(t) \rVert^2 + 2 \lVert \xc(t) \rVert \lVert PE\rVert_{\ui{N}} \notag \\
        &\leq &&-\lVert \xc(t) \rVert^2 + \eta \lVert \xc(t) \rVert^2 + \lVert PE\rVert_{\ui{N}}^2 /\eta\notag  \\
        &\leq &&(\eta-1)/\lambda_\mathrm{max}(P) \cdot V(\xc(t)) + \lVert PE\rVert_{\ui{N}}^2 /\eta,\notag 
    \end{alignat}
    for any $\eta \in (0,1)$ using Young's inequality for products. $\lambda_\mathrm{max}(P)>0$ denotes the smallest eigenvalue of $P$. Let $\alpha \coloneqq (\eta -1)/\lambda_\mathrm{max}(P)$ and $C \coloneqq \lVert PE\rVert_{\ui{N}}^2 /\eta$. As in the proof of Theorem \ref{thm:stability_linear_scalar} (see \eqref{eq:thm_stab_lin_3}) it follows that
    \begin{equation}
        V(\xc(t)) \leq V(\xc(0)) e^{\alpha t} - C/\alpha \cdot(1-e^{\alpha t}). \notag 
    \end{equation}
    This implies that
    \begin{equation}
        \lVert \xc(t) \rVert^2 \leq \frac{\lambda_\mathrm{max}(P)}{\lambda_\mathrm{min}(P)} \lVert \xc(0) \rVert^2 e^{\alpha t} - \frac{C (1-e^{\alpha t})}{\alpha \lambda_\mathrm{min}(P)}. \notag 
    \end{equation}
    To minimise $-C/(\alpha \, \lambda_\mathrm{min}(P))$ as before, we calculate that
    \begin{equation}
        - \frac{C}{\alpha \lambda_\mathrm{min}(P)} 
        = \frac{\lambda_\mathrm{max}(P)}{\lambda_\mathrm{min}(P)} \cdot \frac{C}{1-\eta}=\frac{\kappa(P) \cdot\lVert PE \rVert_{\ui{N}}^2}{\eta(1-\eta)}, \notag 
    \end{equation}
    which is minimised by $\eta = 1/2$ with $-C/(\alpha \lambda_\mathrm{min}(P))=4\lVert PE \rVert_{\ui{N}}^2$, and $\alpha = -1/(2\lambda_\mathrm{max}(P))$. We conclude that
    \begin{align}
        \lVert x(t) \rVert 
        \leq &\lVert x(0) \rVert \sqrt{\kappa(P)} \cdot e^{-\lambda^{-1}_\mathrm{max}(P)\cdot t/4} \notag \\
        &+ 2 \sqrt{\kappa(P)}\lVert PE \rVert_{\ui{N}} + \lVert B \rVert_{\ui{N}}. \notag 
    \end{align}
\end{proof}

For $K=1$, we recover Theorem \ref{thm:stability_linear_scalar} with $P=1/(2(b-a))$.
Next, we generalise Corollary \ref{cor:stability_improved} to the multivariable case under the assumption $\lambda_i = \lambda$ for all $i =1,\dots,N$. As in the scalar case, this simplifies the ultimate bound.

\begin{cor}[Simplified bound for Theorem \ref{thm:stability_multi}]
\label{cor:stab_improv_mult}
    If the assumptions of Theorem \ref{thm:stability_multi} hold, $\lambda_i = \lambda \geq0$ for all $i \in \{1,\dots,N\}$, and there exists $P\succ0$ satisfying \eqref{eq:thm:stability_multi_2} and $\frac{1}{2\lambda} I -P \succeq 0 $, then the system is globally practically exponentially stable,
    \begin{align}
        \Vert x(t) \rVert \leq &\lVert x(0) \rVert\sqrt{\kappa(P)} \cdot e^{-\lambda^{-1}_\mathrm{max}(P)\cdot t/4} \notag \\
        &+ \left(1+\sqrt{\kappa(P)}\cdot(1+2\lVert S \rVert) \right)\lVert B \rVert_{\ui{N}}, \notag 
    \end{align}
    where $S \coloneqq \frac{1}{2}\big( P(A-K) - (A-K)^\intercal P \big)$.
\end{cor}

\begin{proof}
    First, we see that
    \begin{align}
        &2 \lVert P\big( (A-K)B + B\Lambda \big) \rVert_{\ui{N}} \notag \\
        \leq &2 \lVert P(A-K) + \lambda P \rVert \cdot \lVert B \rVert_{\ui{N}}.\notag 
    \end{align}
    Moreover, we have that $2P(A-K) +2\lambda P= -I + 2S +2\lambda P$ by definition of $S$ and \eqref{eq:thm:stability_multi_2}. This implies using $\frac{1}{2\lambda}I - P \succeq 0 $,
    \begin{equation}
        2 \lVert P(A-K) + \lambda P \rVert \leq 2\lVert S \rVert + \lVert I - 2\lambda P \rVert \leq 2\lVert S \rVert +1.\notag 
    \end{equation}
    Inserting both inequalities into \eqref{eq:thm_stability_multi_1} concludes the proof.
\end{proof}

The matrix $S$ is skew-symmetric, i.e., $S^\intercal = -S$, and can be interpreted as representing rotations in the dynamics given by $A-K$ with respect to the metric induced by $P$. Indeed, we can write $P(A-K) = -\frac{1}{2}I + S$. The ultimate bound grows as rotations increase. 

The linear case shows that $x_c$ evolves like an analogue system with uniformly bounded disturbances, see \eqref{eq:thm_stab_mult_1}. By considering the nonlinear case next, we will identify how exactly the dynamics are disturbed in general. To do so, we consider the analogue closed-loop system given by
\begin{align}
    &\dot{ y}(t) = f\big( y(t) + d_1(t)\big) +  u(t) + d_2(t), \label{eq:dyn_sys_disturbed} \\ 
    \text{with} \ \   &u(t) = k( y(t) + d_1(t)) , \notag 
\end{align}
where $d_1(t), d_2(t) \in \R^K$ are disturbances to the system. Inspired by \cite{tabuada2007event}, we define the Input-to-State stability (ISS) of this system in terms of Lyapunov functions (cf. \cite{Sontag2008Input}):

\begin{definition}[Input-to-State stability, adapted from \cite{tabuada2007event}]
\label{def:ISS}
    A smooth function $V\colon \R^K \to \R_{\geq 0}$ is said to be an ISS Lyapunov function for the closed-loop system \eqref{eq:dyn_sys_disturbed} if there exist class $\mathcal{K}_\infty$ functions $\bar \alpha, \ubar{\alpha}, \alpha, \gamma_1$, and $\gamma_2$ satisfying
    \begin{align}
        &\ubar{\alpha}(\lVert y \rVert)\leq V(y) \leq \bar \alpha (\lVert y \rVert), \notag \\
        &\nabla V(y) \cdot (f(y+d_1) + k(y+d_1)+d_2) \notag \\
        \leq &-\alpha (\lVert y \rVert) + \gamma_1(\lVert d_1 \rVert) + \gamma_2(\lVert d_2\rVert) . \label{eq:ISS}
    \end{align}
    The closed-loop system \eqref{eq:dyn_sys_disturbed} is said to be ISS with respect to disturbances $d_1,d_2\in \R^K$ if there exists an ISS Lyapunov function for \eqref{eq:dyn_sys_disturbed}. When $d_1 = d_2 = 0$, $V$ becomes a standard Lyapunov function. 
\end{definition}

\begin{thm}[Nonlinear dynamics]
\label{thm:stability_ISS}
    Let $x(t) \in \R^K$ be the state of an event-based impulsive control system as in \eqref{eq:dyn_sys_ebic:1}--\eqref{eq:dyn_sys_ebic:3} with input matrix $B \in \R^{K \times N}$, input function $g \colon \R^K \to \R_{\geq0}^N$, reset matrix $\Theta \in \R^{N \times N}$, leakage matrix $\Lambda \in \R^{N \times N}$, and drift function $f\colon \R^K \to \R^K$. If the analogue system \eqref{eq:dyn_sys_disturbed} with drift function $f $ and $k(y) \coloneqq B\Theta^{-1}g(y)$ is ISS, then the event-based impulsive control system is globally practically asymptotically stable with functions $\beta$ of class $\mathcal{KL}$ and $\eta_1, \eta_2$ of class $\mathcal{K}_\infty$ such that
    \begin{equation}
        \lVert x(t) \rVert \leq \beta(\lVert x(0)\rVert,t) + \eta_1\left(\lVert B \rVert_{\ui{N}} \right) + \eta_2(\lVert B\Lambda\rVert_{\ui{N}}) .
        \label{eq:thm:_ISS_-1}
    \end{equation}
\end{thm}

\begin{proof}
    We define $\xc(t) \coloneqq x + B\Theta^{-1}z(t)$ as before. Then, it holds that
    \begin{align}
        \xcd{\dot}(t) &= \dot x(t) + B\Theta^{-1}\dot z(t) \notag \\
        \overset{\eqref{eq:auxil_var}}&{=} f(x(t)) + B\Theta^{-1}g(x(t)) - B\Theta^{-1}\Lambda z(t) \notag \\
        &= f\left(\xc(t) + d_1(t)\right) + k\left(\xc(t) + d_1(t)\right) + d_2(t),\notag 
    \end{align}
    where $d_1(t) \coloneqq -B\Theta^{-1}z(t)$ and $d_2(t) \coloneqq -B\Theta^{-1}\Lambda z(t)$. By assumption, the system with $y=x_c$ is ISS. Moreover, due to the bound \eqref{eq:bound_y-x} on $B\Theta^{-1}z(t)$ it holds that $\lVert d_1(t) \rVert \leq \lVert B \rVert_{\ui{N}}$ and similarly $\Vert d_2(t)\rVert \leq \lVert B \Lambda\rVert_{\ui{N}}$ for all $t\geq 0$. Lemma \ref{lem:appendix} (see Appendix) implies that there exist functions $\beta$ of class $\mathcal{KL}$ and $\eta_1',\eta_2$ of class $\mathcal{K}_\infty$ such that
    \begin{equation}
        \lVert \xc(t) \rVert \leq \beta(\lVert \xc(0) \rVert,t) + \eta_1'(\lVert B \rVert_{\ui{N}}) + \eta_2(\lVert B\Lambda \rVert_{\ui{N}}),
        \label{eq:thm_ISS_2}
    \end{equation}
    and using $\xc(0) = x(0)$ we see that
    \begin{align}
        &\lVert x(t)\rVert 
        \leq \lVert \xc(t) \rVert + \lVert B\Theta^{-1}z(t) \rVert \notag \\
        \overset{\eqref{eq:thm_ISS_2}}{\leq} &\beta(\lVert x(0) \rVert,t) + \eta_1(\lVert B \rVert_{\ui{N}}) + \eta_2(\lVert B\Lambda \rVert_{\ui{N}}), \notag 
    \end{align}
    where $r \mapsto \eta_1(r) \coloneqq \eta_1'(r) +r$ is of class $\mathcal{K}_\infty$. 
\end{proof}

\subsection{Minimum inter-event times}
\label{subsec:interevent}

The event-triggering via leaky integrate-and-fire dynamics allows for an easy analysis of the inter-event times. However, we can only bound the inter-event times for each neuronal unit individually. 
Recall that for $i \in \{1,\dots,N\}$,
\begin{align}
    s_i(t) &= \delta_{z_i(t)=\theta_i} \quad \text{and} \label{eq:LIF_1} \\
    \dot z_i(t) &= -\lambda_i z_i(t) + g_i(t) - \theta_i s_i(t), \ \ z_i(0)=0.\label{eq:LIF_2}
\end{align}

\begin{lem}[Bounds on inter-event times]
    If there exist $g_+,g_- \in \R_{\geq 0 }$ such that
    \begin{equation}
        g_- \leq g_i(t) \leq g_+
    \end{equation}
    for all $t \geq 0$,
    then it holds for all $j\geq1$ that
    \begin{align}
        t_{j+1}^i - t_{j}^i &\geq -\frac{1}{\lambda_i}\log\left(1-\frac{\theta_i}{g_+}\lambda_i\right) \geq \frac{\theta_i}{g_+} , \notag \\
        t_{j+1}^i - t_{j}^i &\leq -\frac{1}{\lambda_i}\log\left(1-\frac{\theta_i}{g_-}\lambda_i\right) . \notag
    \end{align}
\end{lem}
\begin{proof}
    We denote $t_0 \coloneqq t_{j}^i$ and $t_1 \coloneqq t_{j+1}^i$. From \eqref{eq:LIF_1} and by solving \eqref{eq:LIF_2} for $t \in [t_0,t_1)$ we obtain
    \begin{align}
        \theta_i 
        &= \lim_{s \uparrow t_1} z_i(s) 
        = \int_{t_0}^{t_1} g_i(s)\,e^{-\lambda_i\left(t_1-s\right)}\,\del s \\
        &\leq \int_{t_0}^{t_1} g_+\,e^{-\lambda_i\left(t_1-s \right)}\,\del s
        = \frac{g_+}{\lambda_i}\left( 1-e^{-\lambda_i\left(t_1 - t_0\right)} \right),
    \end{align}
    which implies
    \begin{equation}
        t_1 - t_0 \geq -\frac{1}{\lambda_i}\log\left(1-\frac{\theta_i}{g_+}\lambda_i\right) \geq \frac{\theta_i}{g_+}.
    \end{equation}
    The upper bound can be derived similarly.
\end{proof}
Note that $-\frac{1}{\lambda_i}\log \left(1-\frac{\theta_i}{g_-}\lambda_i\right) = \infty$ if $\lambda_i \geq g_+/\theta_i$, which means that there are no events. Lemma \ref{lem:lower_bound} directly implies the following result:

\begin{thm}[Minimum inter-event times]
    Let $x \in \R^K$ be the state of an event-based impulsive control system as in \eqref{eq:dyn_sys_ebic:1}--\eqref{eq:dyn_sys_ebic:3} with input function $g \colon \R^K \to \R_{\geq 0}^N$. Let $\theta_\mathrm{min}$ denote the smallest threshold and $\lambda_\mathrm{min}$ the smallest leak constant. If there exists $C\geq0$ such that $\lVert x(t) \rVert \leq C$ for all $t\geq 0$ and $\alpha \in \mathcal{K}_\infty$ such that $g_i(x) \leq \alpha(\lVert x\rVert)$ for all $i \in \{1,\dots,N\}$, then it holds that
    \begin{align}
        &\inf\{t_{j+1}^i-t_j^i \mid j\geq0, 1 \leq i \leq N\} \\
        \geq &-\frac{1}{\lambda_\mathrm{min}}\log\left( 1-\frac{\theta_\mathrm{min}\lambda_\mathrm{min}}{\alpha(C)} \right) \geq \frac{\theta_\mathrm{min}}{\alpha(C)}.
    \end{align}
\end{thm}
\begin{proof}
    It holds $\lVert x(t) \rVert \leq C$ for all $t\geq0$, which implies 
    \begin{equation}
        g_i(x(t)) \leq \alpha(\lVert x(t) \rVert) \leq \alpha(C),
    \end{equation}
    for all $i \in \{1,\dots,N\}$ and $t \geq 0$. Lemma \ref{lem:lower_bound} hence implies
    \begin{align}
        t_{j+1}^i - t_{j}^i &\geq -\frac{1}{\lambda_i}\log\left(1-\frac{\theta_i}{\alpha(C)}\lambda_i\right) \notag \\
        &\geq -\frac{1}{\lambda_\mathrm{min}}\log\left(1-\frac{\theta_\mathrm{min}}{\alpha(C)}\lambda_\mathrm{min}\right) \geq \frac{\theta_\mathrm{min}}{\alpha(C)},
    \end{align}
    which concludes the proof.
\end{proof}

This theorem can be applied to the results of Section \ref{subsec:stability} to obtain lower bounds on the inter-event times.

\subsection{Extension to connected neuronal units}
\label{subsec:extension}

So far, we have assumed that there are no interactions between neuronal units. 
While this simplifies the analysis, larger numbers of neuronal units are usually connected to build networks, so-called spiking neural networks.
In the following, we show that our method can also be applied to connected neuronal units by adapting a spiking controller from \cite{agliati2025spiking}, where the dynamics were derived from a greedy spiking rule. Although that work did not address stability, such a rule would simplify its analysis. In our adaptation, this advantage is absent, requiring an explicit stability analysis.



Suppose that we want to implement an event-based controller with control input $Bs(t)$ inspired by an analogue controller with control input $k(x(t))$. The neuronal unit $i$ should fire if $B_i$ aligns well with $k(x(t))$. This can be achieved by integrating the projection of $k(x(t))$ onto $B_i$ in the neuronal variable $z_i$. Unless all $B_i$ are orthogonal,  this will lead to
redundant spikes and further allow for negative-valued neuronal variables. Both problems can be solved by introducing suitable connections between units. We define the neuronal variables $z(t)$ and spike trains $s(t)$ as follows,
\begin{align}
    s_i(t) &= \delta_{z_i(t) \geq B^\intercal_i B_i} \quad \forall i\in \{1,\dots, N\}, \label{eq:dyn_sys_ebic:4} \\
    \dot z(t) &= - \Lambda z(t) + B^\intercal k(x(t)) - B^\intercal B s(t) , \label{eq:dyn_projection}
\end{align}
where we deliberately added a leakage. The off-diagonal elements of the reset matrix $B^\intercal B$ have two effects. If $B_i$ and $B_j$ point in the same direction, $B_i^\intercal B_j > 0$, a spike of $s_i$ will inhibit $z_j$. If $B_i$ and $B_j$ point in opposite directions, $B_i^\intercal B_j < 0$, a spike of $s_i$ will excite $z_j$, preventing the neuronal variables from becoming indefinitely negative.

We assume that the controller can steer in all directions,
\begin{equation}
    \left\{Br \mid r \in \R_{\geq 0}^N\right\} = \R^K .
    \label{eq:full_row_rank}
\end{equation}
It suffices to have $K+1$ neuronal units: If $(B_i)_{i \in \{1,\dots,K\}}$ is a basis of $\R^K$, one can choose $B_{K+1}=-\sum_{i=1}^KB_i$.
$BB^\intercal$ is invertible due to \eqref{eq:full_row_rank}, and we define similarly to before
\begin{equation}
    \xc(t) = x(t) + (B B^\intercal)^{-1}Bz(t).
\end{equation}
This implies with $R \coloneqq (B B^\intercal)^{-1}B$ and using that $R B^\intercal = I$,
\begin{align}
    &\xcd{\dot}(t) = \dot x(t) + R \dot z(t)  \\
    = &f(x(t)) + Bs(t) - R\Lambda z(t) + R B^\intercal k(x(t)) - R B^\intercal B s(t)  \\
    = &f(x(t)) + k(x(t)) - R\Lambda z(t)  \\
    = &f(\xc(t) + d_1(t)) + k(\xc(t) + d_1(t)) + d_2(t),
\end{align}
with $d_1(t) \coloneqq -Rz(t)$ and $d_2(t) \coloneqq -R\Lambda z(t)$. The same analyses as in Section \ref{subsec:stability} can now be applied if $d_1(t)$ and $d_2(t)$ are uniformly bounded. As intended, the corresponding analogue control input is given by $ k(x)$.

We now show that $z(t)$ is uniformly bounded from above and below, which implies the boundedness of $d_1(t)$ and $d_2(t)$. It still holds $z_i(t) \leq B_i^\intercal B_i$ by definition. To bound $z_i(t)$ from below, we use the following lemma:

\begin{lem}[Bounds on $w^\intercal z(t)$]
\label{lem:lower_bound}
    Let $z(t)$ be the neuronal variables of the event-based impulsive controller defined by \eqref{eq:dyn_sys_ebic:1}, \eqref{eq:dyn_sys_ebic:4}, and \eqref{eq:dyn_projection}.
    \begin{enumerate}
        \item If $\Lambda = \lambda I$ for some $\lambda \geq 0$, then there exists a vector $w \in \R^N$ with strictly positive entries such that for all $t \geq 0$ it holds 
        \begin{equation}
             w^\intercal z(t) = 0. \label{eq:lem_lower_bound_case_1}
        \end{equation}
        \item If $\lambda_i > 0$ for all $i\in \{1,\dots,N\}$, and $\lambda_\mathrm{max} > \lambda_\mathrm{min} $, where $\lambda_\mathrm{max}, \lambda_\mathrm{min}$ denote the largest and smallest $\lambda_i$, respectively, then there exists a vector $w \in \R^N$ with strictly positive entries such that for all $t \geq 0$ it holds 
        \begin{equation}
             w^\intercal \mathrm{diag}(B^\intercal B) \geq w^\intercal z(t) \geq \gamma \,w^\intercal \mathrm{diag}(B^\intercal B), \label{eq:lem_lower_bound_case_2}
        \end{equation}
        for $\gamma \coloneqq \left( 1- \frac{\lambda_\mathrm{max}}{\lambda_\mathrm{min}} \right)<0$.
    \end{enumerate}
\end{lem}
\begin{proof}
    There exist $q_i \in \R_{\geq 0}^N$ such that $Bq_i = -B_i$ due to \eqref{eq:full_row_rank} for all $i \in \{1,\dots N\}$. Then, let $w \coloneqq \sum_{i=1}^N e_i + q_i$, where $e_i$ is the unit vector in direction $i$. This yields $w_i \geq 1 $ and $Bw = \sum_{i=1}^N Be_i + Bq_i = 0$.
    Now, we obtain
    \begin{align}
        w^\intercal \dot z(t) &= -w^\intercal \Lambda z(t) +  w^\intercal B^\intercal k(x(t)) - w^\intercal B^\intercal B s(t) \notag \\
        &= - w^\intercal \Lambda z(t). \notag
    \end{align}
    
    \textit{1)} If $\Lambda = \lambda I$ for $\lambda \geq0 $, this implies using $z(0)=0$ that
    \begin{equation}
        w^\intercal z(t) = w^\intercal z(0) \cdot e^{-\lambda t} = 0, \notag
    \end{equation}
    which proves \eqref{eq:lem_lower_bound_case_1}.

    \textit{2)} If $\lambda_i > 0$ for all $i\in \{1,\dots,N\}$, and $\lambda_\mathrm{max} > \lambda_\mathrm{min} $, we need to relate $w^\intercal \Lambda z(t)$ to $w^\intercal z(t)$. We denote the positive and negative part of $z(t)$ by $z_+(t)$ and $z_-(t)$, respectively, such that $z(t) = z_+(t) - z_-(t)$. Moreover, let $d \coloneqq \mathrm{diag}(B^\intercal B)$. Then, it holds
    \begin{align}
        w^\intercal \Lambda z(t) &= w^\intercal \Lambda z_+(t) - w^\intercal \Lambda z_-(t) \notag \\
        &\leq \lambda_\mathrm{max}w^\intercal z_+(t) - \lambda_\mathrm{min} w^\intercal z_-(t) \notag \\
        &\leq (\lambda_\mathrm{max} - \lambda_\mathrm{min}) w^\intercal d + \lambda_\mathrm{min} w^\intercal z(t), \notag
    \end{align}
    where we used $z_+(t) \leq d$ elementwise. This implies
    \begin{alignat}{3}
        &&&w^\intercal \dot z(t) &&= -w^\intercal \Lambda z(t) \notag  \\
        &&&&&\geq (\lambda_\mathrm{min} - \lambda_\mathrm{max}) w^\intercal d - \lambda_\mathrm{min} w^\intercal z(t) \overset{(*)}{\geq} 0 \notag \\
        &\iff &&w^\intercal z(t) &&\leq \left(1- \frac{\lambda_\mathrm{max}}{\lambda_\mathrm{min}} \right)\cdot w^\intercal d = \gamma \, w^\intercal d , \notag
    \end{alignat}
    where the equivalence is with respect to the inequality marked by $(*)$. This implies with $w^\intercal z(0) = 0$ that $w^\intercal z(t) \geq \gamma \, w^\intercal \mathrm{diag}(B^\intercal B)$. The upper bound follows directly from $z_i(t) \leq B_i^\intercal B_i$.
\end{proof}

Let us consider the case $\Lambda = \lambda I$ for some $\lambda \geq 0$. Lemma \ref{lem:lower_bound} shows that a linear combination of the neuronal variables with strictly positive weights $w$ exists such that $w^\intercal z(t) = 0$ for all $t \geq 0$. Since $z(t)$ is uniformly bounded from above elementwise, we can conclude that $z(t)$ is also bounded from below elementwise. If $z_i(t) \geq 0$, the bound is trivial. If $z_i(t) < 0$, we have
\begin{align}
    z_i(t) &= \frac{1}{w_i} w_i z_i(t) \geq \sum_{j:z_j(t)<0} \frac{w_j z_j(t)}{w_i} \notag \\
    &= -\sum_{j\colon z_j(t) \geq0} \frac{w_j z_j(t)}{w_i} 
    \geq - \sum_{j:j\not = i} \frac{w_j B^\intercal_j B_j}{w_i} , \notag
\end{align}
where we used Lemma \ref{lem:lower_bound}, and that $z_j(t) \leq B^\intercal_j B_j$. If $\lambda_i > 0$ for all $i\in \{1,\dots,N\}$, and $\lambda_\mathrm{max} > \lambda_\mathrm{min} $, it follows similarly that
\begin{equation}
    z_i(t) \geq  \frac{\gamma \, w^\intercal \mathrm{diag}(B^\intercal B)}{w_i} - \sum_{j:j\not = i} \frac{w_j B^\intercal_j B_j}{w_i}. \notag
\end{equation}
In both cases, this implies that $d_1(t)$ and $d_2(t)$ are uniformly bounded.


\section{Numerical experiments}
\label{sec:numerics}

We simulate the event-based impulsive controllers introduced in the previous sections to illustrate our theoretical results. First, we consider the scalar linear case of Theorem \ref{thm:stability_linear_scalar} using the controller defined in Remark \ref{rem:literature} with $x(0)=2.6$. For $a=1$, $\alpha=1$, and $\theta=1/2.5$, we obtain $b=2.5$, and hence the condition $b\geq a + \lambda$ is not fulfilled for $\lambda = 3$. In this case, Fig.~\ref{fig:2_1dim} (left) shows that $x(t)$ never reaches zero due to the leakage, and $x(t)$ is indeed lower than the bound \eqref{eq:thm_stab_lin_0}. If $\lambda=1.5$, \eqref{eq:cor_lin_2} of Corollary~\ref{cor:stability_improved} applies, and we can see from Fig.~\ref{fig:2_1dim} (right) that $x(t)$ oscillates around zero and is more tightly bound by \eqref{eq:cor_lin_2}. The middle panel of Fig.~\ref{fig:2_1dim} shows that the bound is tight for $\lambda = 1.5$, i.e., when $b=a+\lambda$.

Next, we analyse whether the condition $a<b$ is necessary for stability. We consider the same controller as before with $\lambda\in\{0,3\}$, $a\in [0,5]$ and $\theta \in [1/5,\infty)$, such that $b \in (0,5]$. We initialise the plant at $x(0)=50.5$ and simulate for time $T=200$. If the control system is stable and $x(0)$ is far away from the ultimate bound, it should hold $x(0) > x(T)$, and otherwise it should hold $x(0) < x(T)$. We choose $C=\log\big(|x(T)/x(0)|\big)/T$ as a heuristic stability measure, approximating the Lyapunov exponent. Fig.~\ref{fig:2_heatmap} shows that the system is stable, $C<0$, if and only if $a<b$. At $a=b$, \eqref{eq:thm_stab_lin_2} shows that the system is only stable if $\lambda=0$. Since a positive leakage only induces a linear growth rather than an exponential growth, this difference is not visible in Fig.~\ref{fig:2_heatmap}.

\begin{figure}[htbp]
\centerline{\includegraphics[width=0.48\textwidth]{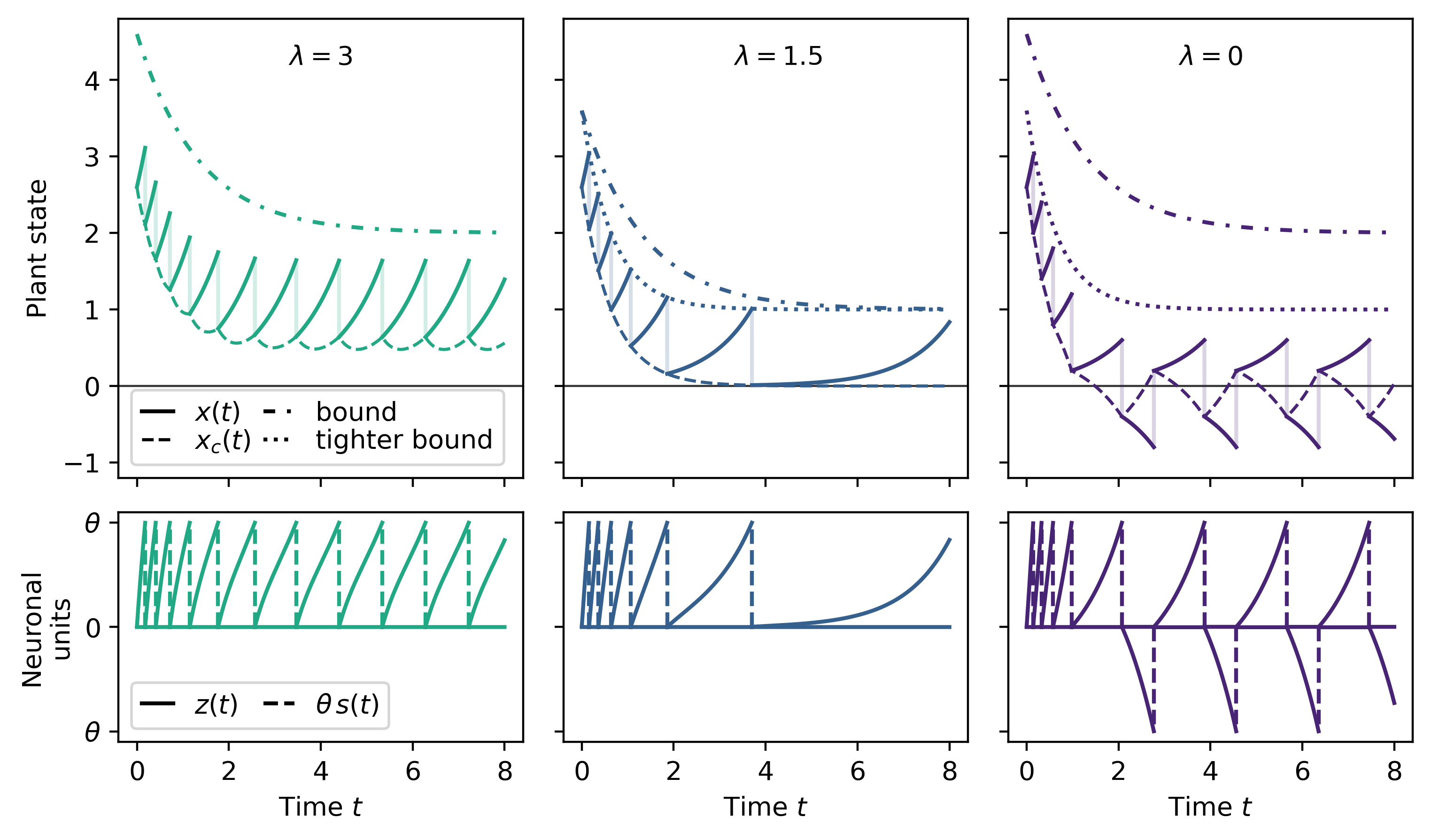}}
\caption{
A one-dimensional linear plant controlled by an event-based impulsive controller. The state $x(t)$ (solid) with discontinuities (shaded), the continuous variable $x_c(t)$ (dashed), the bound \eqref{eq:thm_stab_lin_0} (dash-dotted), and the tighter bound \eqref{eq:cor_lin_2} (dotted) are shown in the upper plots for $\lambda=3$ (left, cyan), $\lambda = 1.5$ (middle, blue), and $\lambda=0$ (right, purple). The two neuronal variables $z(t)$ (solid), and the rescaled spikes $\theta s(t)$ (dashed) are shown in the lower plots (units 1 and 2 in the top and bottom half, respectively). The control system is given by $f(x) = x$, $B=[-1,1]$, $\theta=1/2.5$, and $g(x) = \big[[x]_+, [-x]_+\big]^\intercal$.
}
\label{fig:2_1dim}
\end{figure}

\begin{figure}[htbp]
\centerline{\includegraphics[width=0.48\textwidth]{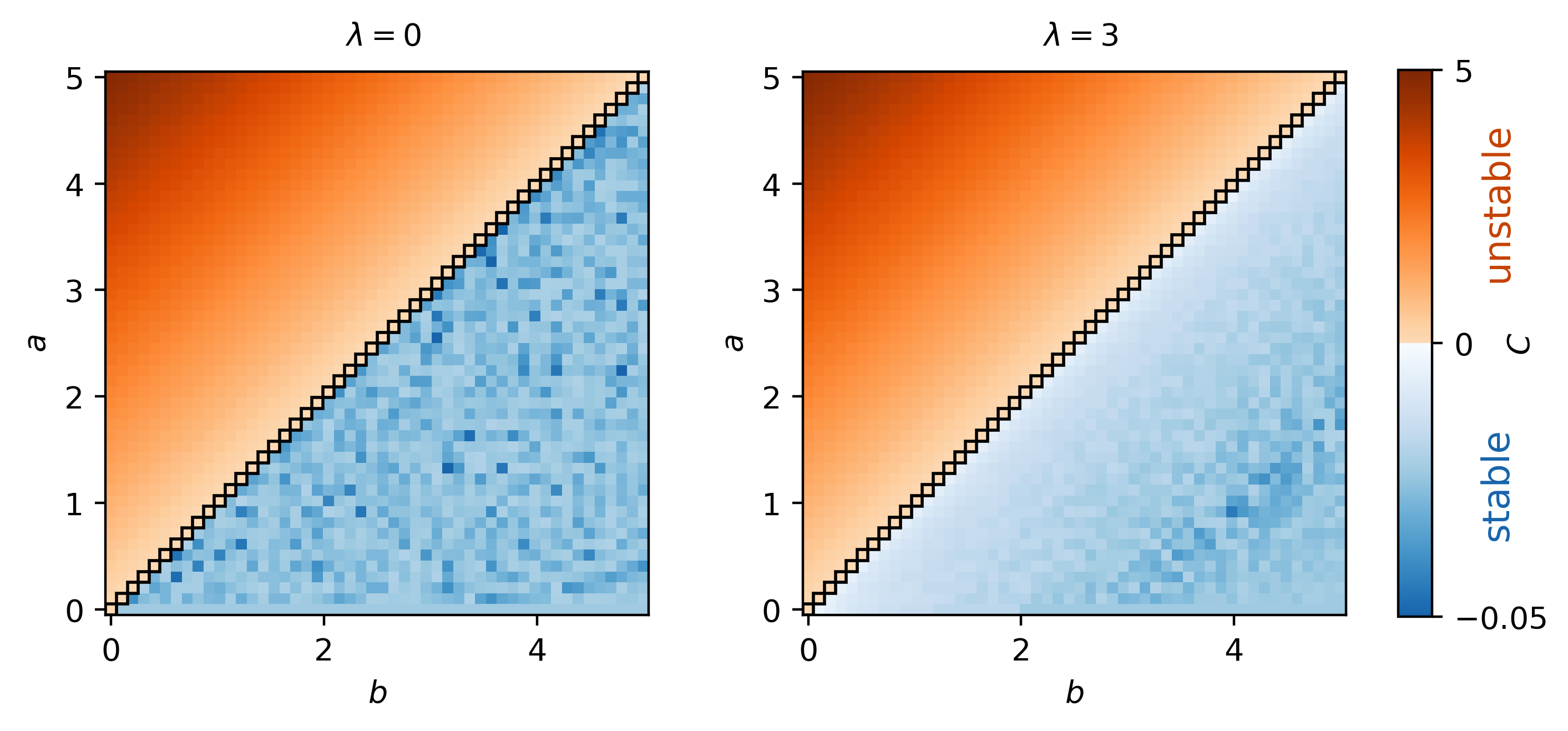}}
\caption{
Heatmap of the stability of the event-based impulsive control system as a function of $a$ and $b$. We consider $f(x)=ax$, $B/\theta\cdot g(x) = -bx$, $\lambda = 0$ (left) and $\lambda = 3$ (right), and $B,g$ as in Fig. \ref{fig:2_1dim}. The stability for a pair $a,b\in[0,5]$ is measured by $C=\log\big(| x(T)/x(0)|\big)/T$ for $x(0)=50.5$ and $T=200$. The black squares on the diagonal indicate $a=b$.
}
\label{fig:2_heatmap}
\end{figure}

In Corollary \ref{cor:stab_improv_mult}, we obtained an additional term compared to the scalar case that depends on rotations. To illustrate its impact, we consider a two-dimensional state with the same dynamics as before on the radial axis and additional rotation on the angular axis modulated by $\omega$, $f(x) = \left(\begin{smallmatrix}
    1 & \omega \\
    -\omega & 1
\end{smallmatrix}\right)\,x$. For the controller, we use the previous controller for the two Cartesian axes $x_0$ and $x_1$, respectively, which results in $N=4$ neuronal units:
$
    B=
    \left(\begin{smallmatrix}
    -1 & 1 & 0 & 0\\
    0 & 0 & -1 & 1
    \end{smallmatrix}\right)$, $g(x) = \big[[x_0]_+, [-x_0]_+, [x_1]_+, [-x_1]_+\big]^\intercal
$. In Fig. \ref{fig:3_2dim}, we simulate the control system for $a=1$, $\lambda=0.2$ and $\theta=1/1.5$, such that Corollary \ref{cor:stab_improv_mult} applies with $\lVert S \Vert = \omega/(2(b-a))$. For both weak and strong rotations, $\omega \in \{0.5, 3\}$, the state $x(t)$ converges within the ball with radius given by the ultimate bound. As $\omega$ increases, the ultimate bound becomes more conservative.

\begin{figure}[htbp]
\centerline{\includegraphics[width=0.48\textwidth]{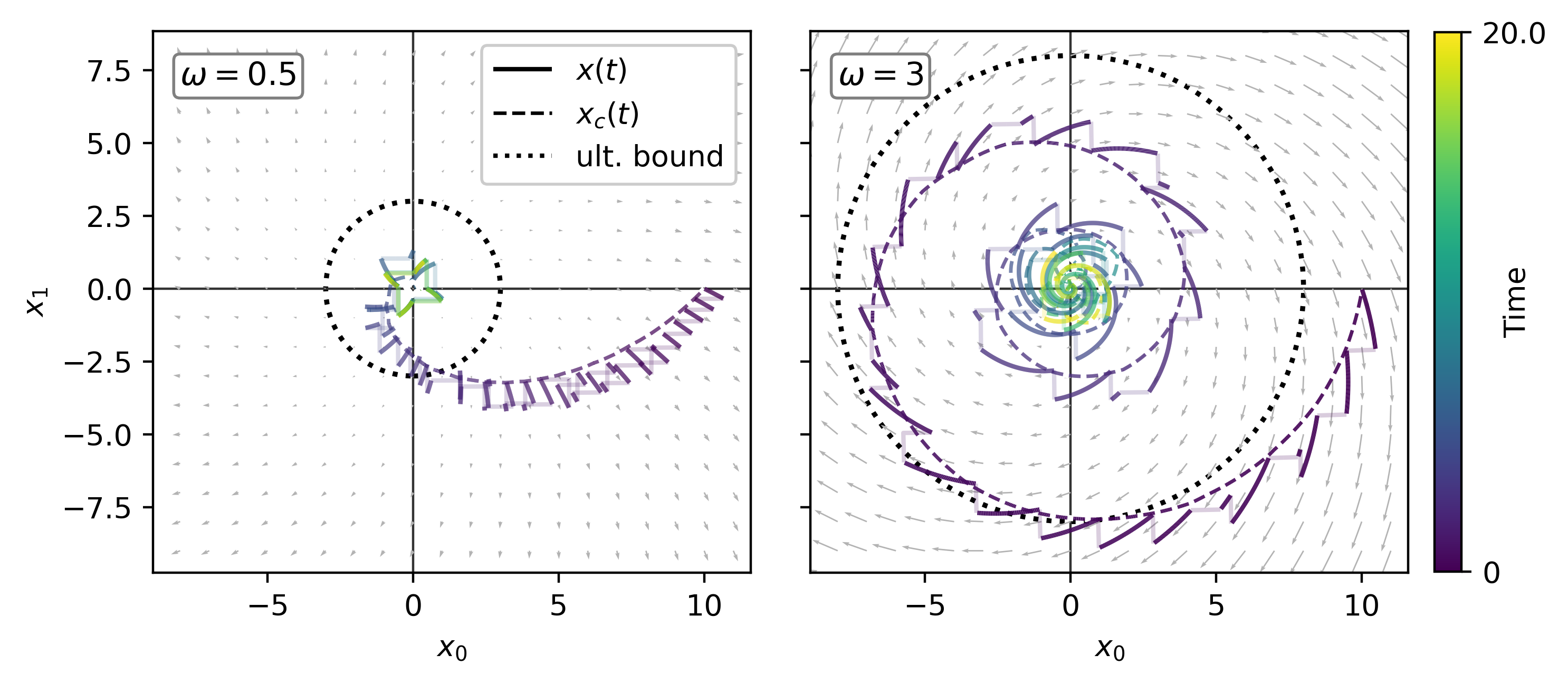}}
\caption{
A two-dimensional linear plant controlled by an event-based impulsive controller. The state $x(t)$ (solid) with discontinuities (shaded), the continuous variable $x_c(t)$ (dashed), and a ball with radius given by the ultimate bound (dotted) are shown for $\omega = 0.5$ (left) and $\omega=3$ (right). The control system is given by $f(x) = \left(\begin{smallmatrix}
    1 & \omega \\
    -\omega & 1
\end{smallmatrix}\right)\,x$, $B=\left(\begin{smallmatrix}
    -1 & 1 & 0 & 0\\
    0 & 0 & -1 & 1
\end{smallmatrix}\right)$, $g(x) = \big[[x_0]_+, [-x_0]_+, [x_1]_+, [-x_1]_+\big]^\intercal$, $\lambda=0.2$ and $\theta=1/1.5$.
}
\label{fig:3_2dim}
\end{figure}

\section{Conclusion}
\label{sec:conclusion}

We have analysed the stability and inter-event time properties of an event-based impulsive controller with neuronal dynamics. By leveraging the explicit dynamics of the plant and neuronal variables, we introduced an auxiliary variable free of discontinuities. This revealed a fundamental link between the event-based impulsive controller and a corresponding analogue controller, allowing us to apply standard arguments for analogue control theory. 
Using this novel method, we proved global practical asymptotic stability for nonlinear systems and global practical exponential stability for linear ones, and derived lower bounds for the inter-event times.

From a control theory perspective, the current work can be extended in the following ways. First, we expect our results to 
generalise to variable target states.
Second, the controller
only depends on the current error, rather than its history or derivative. One could consider causal operators instead. Finally, one could explore whether the stability of the analogue control system is not only sufficient but also necessary for the practical stability of the event-based impulsive control system, as our simulations suggest.


The presented model can also be extended in various directions from a biological perspective. Although we consider multiple neuronal units, their connectivity is constrained.
The class of connectivities to which our method is applicable remains to be identified.
Further, spikes could be abstracted by kernel convolutions instead of simple impulses. 
We expect our framework to extend naturally to this case, with the corresponding analogue controller sharing the same convolutional dynamics. Finally, many neuron models entail mechanisms for time regularisation, such as refractory periods or spike frequency adaptation, which pose interesting challenges for the stability analysis.

We here established a connection between event-based impulsive and analogue controllers in terms of stability. This relation lays the foundation for a systematic mapping between the two control types, extending beyond stability. Such a mapping will depend on neuronal properties such as leakage, refractoriness, and spike frequency adaptation. Importantly, such a framework could clarify the advantages of event-based impulsive control over analogue control with respect to plant dynamics and control costs.

\appendix

\begin{lem}[Special case of the comparison lemma]
\label{lem:comparison}
    Let $v(t) \colon [0,\infty) \to \R$ be a continuous and differentiable function and constants $\alpha, \beta \in \R$ with $\alpha \not= 0$ such that
    \begin{equation}
        \dot v(t) \leq \alpha v(t) + \beta \qquad \text{for all } t\geq0. \notag
    \end{equation}
    Then, it holds that
    \begin{equation}
        v(t) \leq \left( v(0) + \frac{\beta}{\alpha} \right) \cdot e^{\alpha t} - \frac{\beta}{\alpha} \qquad \text{for all } t\geq0. \notag
    \end{equation}
\end{lem}
\begin{proof}
    Let us consider
    \begin{equation}
        w(t) \coloneqq \left( v(0) + \frac{\beta}{\alpha} \right) \cdot e^{\alpha t} - \frac{\beta}{\alpha} . \notag
    \end{equation}
    It follows that $w(t)$ satisfies $w(0)=v(0)$ and
    \begin{equation}
        \dot w(t) = \alpha w(t) + \beta . \notag
    \end{equation}
    By the comparison lemma (Lemma 3.4 of \cite{khalil2002nonlinear}), we conclude that
    \begin{equation}
        v(t) \leq w(t) = \left( v(0) + \frac{\beta}{\alpha} \right) \cdot e^{\alpha t} - \frac{\beta}{\alpha}. \notag
    \end{equation}
\end{proof}

\begin{thm}[Tighter bound for Theorem \ref{thm:stability_linear_scalar}]
\label{thm:stability_improved_appendix}
    Let the assumptions of Theorem \ref{thm:stability_linear_scalar} hold, and let $g(x)=(g_1(x),g_2(x))$. If $b \geq a + \lambda$, and if $g_1(x)=0$ if $x\leq0$, and $g_2(x)=0$ if $x\geq0$, then the event-based impulsive control system is globally practically exponentially stable,
    \begin{equation}
        | x(t) | \leq | x(0) | \cdot e^{(a-b)\cdot t} + \lVert B \rVert_{\ui{2}}.
        \label{eq:cor_lin_3}
    \end{equation}
\end{thm}
\begin{proof}
    Let $z(t)=(z_1(t),z_2(t))$. We first consider the assumptions on the input function $g$. If $x(0)=0$, it follows $x(t)=0$ for all $t\geq0$ since $g_1(0)=g_2(0)=0$. We denote the first event time by $t_0$ and consider $x(0)>0$. We assume right-continuity of all variables. Due to the uncontrolled dynamics given by $f(x)=ax$, it follows $x(t)>0$ for all $t \in [0,t_0)$, and hence $z_2(t)=0$ for all $t\in[0,t_0]$ due to $z(0)=(0,0)$. At $t_0$, it again holds $z(t_0)=(0,0)$ and either $x(t_0)>0$ or $x(t_0)\leq0$. The same reasoning holds for $x(0)<0$, but in this case we have that $z_1(t)=0$ for all $t\in[0,t_0]$. Hence, the neuronal units are only sequentially active, depending on the sign of $x$ after the last impulse.

    Next, we consider the input function directly. Let $B=(b_1, b_2)$. This implies $\lVert B \rVert_{\ui{2}} = \max(|b_1|,|b_2|)$ if $b_1$ and $b_2$ have opposite signs, which is always the case as we will show. By assumption, it holds that
    \begin{align}
        -bx &= B\Theta^{-1}g(x) \notag \\
        &= \frac{1}{\theta} \big(b_1g_1(x)+b_2g_2(x) \big) = \begin{cases}
            \frac{b_1}{\theta} g_1(x) \quad &\text{if } x\geq0, \\
            \frac{b_2}{\theta}g_2(x) \quad &\text{if } x \leq 0.
        \end{cases} \notag
    \end{align}
    If $b_1=0$ or $b_2=0$, it follows $b=0$ and $x(t)=x(0)\,e^{at}$, which trivially satisfies \eqref{eq:cor_lin_3}. Hence, let $b_1,b_2 \not = 0$, which implies
    \begin{equation}
        g_i(x) = \begin{cases}
            -\frac{b\theta}{b_1} [x]_+ &\text{if }i=1, \\
            \frac{b\theta}{b_2} [-x]_+ &\text{if } i=2.
        \end{cases} \notag
    \end{equation}

    \textit{Case 1, $b<0$:} Since $g_i(x) \geq 0$ for $i\in \{1,2\}$, it holds $\frac{b\theta}{b_1} \leq 0$ and $\frac{b\theta}{b_2} \geq 0$, which yields $b_1>0$ and $b_2<0$. Hence, the first neuronal unit, which is only active if the state is strictly positive, leads to further increases of the state, and conversely, the second neuronal unit, which is only active if the state is strictly negative, leads to further decreases of the state. In consequence, the state will never change its sign if $b<0$. Let us hence assume that $x(t)>0$ for all $t\geq 0$. This implies $z_2(t) = 0$ for all $t\geq0$, and thus
    \begin{equation}
        x_c(t) = x(t) + b_1 z_1(t)/\theta \geq x(t). \notag
    \end{equation}
    Further, it implies $B\Theta^{-1}z(t) = b_1 z_1(t)/\theta > 0$ for $t \geq 0$. If we insert this into \eqref{eq:thm_stab_lin_1}, we obtain
    \begin{align}
        \xcd{\dot}(t) &= (a-b)x_c(t) - (a-b+\lambda) b_1 z_1(t)/\theta \notag \\
        &\leq (a-b)x_c(t) + (b-a-\lambda) b_1, \label{cor:proof:ext5}
    \end{align}
    using $b-a-\lambda \geq 0$ and $b_1 > 0$. The comparison lemma (Lemma \ref{lem:comparison}) now implies that
    \begin{align}
        \xcd{}(t) &\leq \frac{b-a-\lambda}{b-a} b_1 + \left( \xcd{}(0) - \frac{b-a-\lambda}{b-a} b_1 \right)\cdot e^{(a-b)t} \notag \\
        &\leq x_c(0)\cdot e^{(a-b)t} + b_1. \notag
    \end{align}
    Using $\xcd{}(0)=x(0)$ and $x_c(t) \geq x(t)$, this yields
    \begin{equation}
        x(t) \leq x(0)\cdot e^{(a-b)t} + \lVert B \rVert_{\ui{2}} . \notag
    \end{equation}
    Similar arguments can be applied if $x(t)<0$ for all $t\geq0$.
    
    \textit{Case 2, $b>0$:} Opposite to the case $b<0$, it holds $b_1<0$ and $b_2>0$. Hence, the first neuronal unit, which is only active if the state is strictly positive, leads to decreases of the state, and conversely, the second neuronal unit, which is only active if the state is strictly negative, leads to increases of the state. Let us assume that $x(0)>0$ and let us denote by $t_1$ the first event time where $x(t_1)<0$. This time might not exist, in which case we set $t_1=\infty$. It now holds that $x(t) > 0$ for $t \in [0,t_1)$, and similar to \eqref{cor:proof:ext5} we obtain
    \begin{align}
        \xcd{\dot}(t) 
        &= (a-b)\xcd{}(t) - (a-b+\lambda) b_1 z_1(t) /\theta \notag \\
        &\leq (a-b) x_c(t), \notag
    \end{align}
    for $t \in [0,t_1)$ using $a-b+\lambda \leq 0$. The comparison lemma (Lemma \ref{lem:comparison}) now implies that
    \begin{equation}
        x_c(t) \leq x(0)\cdot e^{(a-b)t} \quad \text{for all } t \in [0,t_1), \notag
    \end{equation}
    where we used $\xcd{}(0) = x(0)$. If $x_c(t)\geq0$, we have $|\xcd{}(t)| \leq |x(0)|\cdot e^{(a-b)t}$, and with \eqref{eq:bound_y-x} this yields
    \begin{equation}
        |x(t)| \leq |x(0)|\cdot e^{(a-b)t} + \lVert B \rVert_{\ui{2}} . \notag
    \end{equation}
    If $\xcd{}(t)<0$, $x(t)$ and $\xcd{}$ have opposite signs, which implies again using \eqref{eq:bound_y-x} that $|x(t)| \leq \lVert B \rVert_{\ui{2}}$. In total, this shows
    \begin{equation}
        |x(t)| \leq |x(0)|\cdot e^{(a-b)t} + \lVert B \rVert_\ui{2} \quad \text{for all } t\in[0,t_1).
    \end{equation}
    This concludes the proof if $t_1 = \infty$. For $t_1<\infty$, we analyse the conditions that enable a sign change of the state.
    
    \textit{Case 2.1, $a<0$:} If the sign of the state changes at time $t_1$, we can conclude that $|x(t_1)| \leq \lVert B \rVert_{\ui{2}} $. Since the uncontrolled dynamics drive the state further towards zero, and the impulses lead to jumps towards zero with overshoots of at most $\max(|b_1|,|b_2|)= \lVert B \rVert_{\ui{2}}$, we can conclude that $|x(t)| \leq \lVert B \rVert_{\ui{2}} $ for all $t\geq t_1$.

    \textit{Case 2.2, $a\geq 0$:}
    We denote by $t_0$ the event time just before $t_1$. If $t_1$ is the first event, we set $t_0=0$. We can solve \eqref{eq:LIF_2} using $z_1(t_0)=0$ for $t\in [t_0,t_1)$, which yields
    \begin{align}
        z_1(t) &= \int_{t_0}^t e^{-\lambda(t-s)} \cdot\frac{b\theta}{-b_1} x(t_0)\cdot e^{as} \,\del s \label{eq:diff_eq_solution_z1} \\
        &=\frac{b\theta}{-b_1(a+\lambda)} x(t_0) \left( e^{a(t-t_0)} - e^{-\lambda(t-t_0)} \right) .\notag 
    \end{align}
    Since $\theta = \lim_{t \uparrow t_1} z_1(t)$, this implies
    \begin{equation}
        \theta = \frac{b\theta}{-b_1(a+\lambda)} x(t_0) \left( e^{a(t_1-t_0)} - e^{-\lambda(t_1-t_0)} \right), \notag
    \end{equation}
    which yields
    \begin{equation}
        x(t_0)\cdot e^{a(t_1-t_0)}  = x(t_0)\cdot e^{-\lambda(t_1-t_0)}+\frac{-b_1(a+\lambda)}{b} .\label{cor:proof:ext1}
    \end{equation}
    Using this, we can express $x(t_1)$ as
    \begin{align}
        x(t_1) &= x(t_0)\cdot e^{a(t_1-t_0)} + b_1 \notag \\
        \overset{\eqref{cor:proof:ext1}}&{=} x(t_0) \cdot e^{-\lambda(t_1-t_0)} + \frac{-b_1(a+\lambda)}{b} + b_1 \notag \\
        &= x(t_0) \cdot e^{-\lambda(t_1-t_0)} + \frac{b_1(b-a-\lambda)}{b} \label{cor:proof:ext2} \\
        &\geq \frac{b_1(b-a-\lambda)}{b} \geq b_1. \label{cor:proof:ext3}
    \end{align}
    Further, we can show a sufficient condition for $x(t_1)\leq 0$:
    \begin{align}
        &x(t_0) \leq \frac{-b_1(b-a-\lambda)}{b}\label{cor:proof:ext4} \\
        \implies &x(t_0)\cdot e^{-\lambda(t_1-t_0)} \leq \frac{-b_1(b-a-\lambda)}{b} \notag \\
        \implies &x(t_1) \overset{\eqref{cor:proof:ext2}}{=} x(t_0) \cdot e^{-\lambda(t_1-t_0)} +  \frac{b_1(b-a-\lambda)}{b} \leq 0 . \notag
    \end{align}
    Let $t_2$ denote the time of the event after the event $t_1$. Repeating the calculations \eqref{eq:diff_eq_solution_z1}--\eqref{cor:proof:ext3} for $t\in[t_1,t_2)$ and $x(t_1)<0$ yields
    \begin{align}
        x(t_2) &= x(t_1) \cdot e^{-\lambda(t_2-t_1)} +\frac{b_2(b-a-\lambda)}{b} \label{cor:proof:ext8} \\
        &\leq \frac{b_2(b-a-\lambda)}{b}. \label{cor:proof:ext6}
    \end{align}
    With a similar approach to \eqref{cor:proof:ext4}, we obtain a sufficient condition for $x(t_2) \geq 0$,
    \begin{equation}
        x(t_1) \geq \frac{-b_2(b-a-\lambda)}{b} \implies x(t_2) \geq 0. \label{cor:proof:ext7}
    \end{equation}

    \textit{Case 2.2.1, $-b_1=b_2$:} If we apply \eqref{cor:proof:ext3}, \eqref{cor:proof:ext7}, \eqref{cor:proof:ext6}, and \eqref{cor:proof:ext4} in this order sequentially, we see that the sign of $x(t)$ changes at every event time after $t_0$, cf. right panel of Fig.~\ref{fig:2_1dim}. Since the impulses have size $|b_1|$, we obtain $|x(t)| \leq |b_1| = \lVert B \rVert_{\ui{2}}$ for $t\geq t_1$. Similar arguments can be applied if $x(0)<0$.

    \textit{Case 2.2.2, $-b_1 \not =b_2$:} Let us assume without loss of generality that $|b_1|>|b_2|$, i.e., $0 > b_1+b_2$, $x(t_0)>0$, and again $x(t_1)<0$. We want to ensure that $|x(t)| \leq |b_1| = \lVert B \rVert_{\ui{2}}$ for all $t \in [t_1,t_2)$, where $t_2$ denotes the next event time such that $x(t_2)>0$. For $t \in [t_1,t_2)$, only the second neuronal unit is active, and we denote by $t_k', t_{k+1}'$ two arbitrary subsequent event times with $t_1 \leq t_k' < t_{k+1}' \leq t_2$. Similar to \eqref{cor:proof:ext8}, it holds
    \begin{equation}
        x(t_{k+1}') = x(t_k')\cdot e^{-\lambda (t_{k+1}'-t_k')} + b_2 \frac{b-a-\lambda}{b} \geq x(t_k'). \notag
    \end{equation}
    By repeating this argument, we obtain that $x(t_k') \geq x(t_1)$.    
    From \eqref{cor:proof:ext3}, we know that $x(t_1) \geq \frac{b_1(b-a-\lambda)}{b}$. Since $a>0$, we obtain for $t\in [t_1,t_2)$ that
    \begin{align}
        x(t) &\geq \max_k \lim_{s \uparrow t_{k+1}'} x(s) 
        = \max_k x(t_k')\cdot e^{a(t_{k+1}'-t_k')} \notag \\
        &= \max_k x(t_k') \cdot e^{-\lambda (t_{k+1}'-t_k')} + \frac{-b_2(a+\lambda)}{b} \notag \\
        &\geq x(t_1) + \frac{-b_2(a+\lambda)}{b} \notag \\
        \overset{\eqref{cor:proof:ext3}}&{\geq} \frac{b_1(b-a-\lambda)}{b} - \frac{b_2(a+\lambda)}{b} \notag \\
        &= b_1 - \frac{(b_1+b_2)(a+\lambda)}{b} \geq b_1, \notag 
    \end{align}
    where we used that $a, \lambda, b\geq 0$, $b_1+b_2<0$, and the second line follows from similar calculations as \eqref{cor:proof:ext1}. By \eqref{cor:proof:ext6}, it holds
    \begin{equation}
        x(t_2) \leq \frac{b_2(b-a-\lambda)}{b} \leq \frac{-b_1(b-a-\lambda)}{b}. \notag 
    \end{equation}
    Hence, condition \eqref{cor:proof:ext4} is fulfilled and the sign of $x(t)$ again changes at the next event time. The same reasoning can then be iterated, and similar arguments can be applied if $x(0)<0$.
\end{proof}

\addtolength{\textheight}{-2.5cm}

\begin{lem}[Stability under bounded disturbances]
\label{lem:appendix}
    If the closed-loop system \eqref{eq:dyn_sys_disturbed} is ISS with respect to disturbances $d_1(t),d_2(t)$, and there exist $D_1,D_2 \geq 0$ such that $\lVert d_1(t)\rVert\leq D_1$ and $\lVert d_2(t) \rVert \leq D_2$ for all $t\geq 0$, then the system is globally practically asymptotically stable, with functions $\beta$ of class $\mathcal{KL}$ and $\eta_1, \eta_2$ of class $\mathcal{K}_\infty$ such that
    \begin{equation}
        \lVert y(t) \rVert \leq \beta(\lVert y(0) \rVert, t) + \eta_1(D_1) + \eta_2(D_2).
    \end{equation}
\end{lem}
\begin{proof}
     Since the system is ISS, there exist class $\mathcal{K}_\infty$ functions $\bar \alpha, \ubar{\alpha}, \alpha, \gamma_1$ and $\gamma_2$ such that for all $t\geq0$,
    \begin{align}
        &\ubar{\alpha}(\lVert y(t) \rVert) \leq V(y(t)) \leq \bar \alpha (\lVert y(t) \rVert), \label{eq:thm_ISS_0}\\
        &\dot V(y(t)) \leq -\alpha (\lVert y(t) \rVert) + \gamma_1(\lVert d_1(t) \rVert) +\gamma_2(\lVert d_2(t)\rVert).
        \label{eq:thm_ISS_1}
    \end{align}
   By assumption, it holds that $\gamma_1(\lVert d_1(t) \rVert) \leq \gamma_1(D_1) \eqqcolon E$, and $\gamma_2(\Vert d_2(t)\rVert) \leq \gamma_2(D_2) \eqqcolon F$ for all $t\geq0$.
   Since the inverse of a strictly increasing function is strictly increasing, \eqref{eq:thm_ISS_0} yields $\bar{\alpha}^{-1}(V(y(t))) \leq \lVert y(t) \rVert$. Using \eqref{eq:thm_ISS_1}, this implies
    \begin{equation}
        \dot V(y(t)) \leq -\alpha\left(\bar\alpha^{-1}(V(y(t)))\right) + E + F
    \end{equation}
    Let $\alpha_1 \coloneqq \alpha \circ \bar\alpha^{-1}$, $C \coloneqq \alpha_1^{-1}(E + F)$, and $W(t) \coloneqq V(y(t)) - C$. We then obtain
    \begin{align}
        \dot W(t) &= \dot V(y(t)) \leq -\alpha_1(V(y(t))) + E + F \\
        &= -\alpha_1\left(W(t) + C\right) + \alpha_1(C) \eqqcolon -\alpha_2(W(t)).
    \end{align}
    Note that $\alpha_2$ is of class $\mathcal{K}_\infty$ since $\alpha_2(0) = -\alpha_1(0+C) + \alpha_1(C)=0$. Let $\beta_1(r,t)$ be the solution to 
    \begin{equation}
        \dot{\widetilde W}(t)=-\alpha_2\left(\widetilde{W}(t)\right), \quad \widetilde W(0)=r.
    \end{equation}
    Then, $\beta_1$ is of class $\mathcal{KL}$ and by the comparison lemma (Lemma 3.4 of \cite{khalil2002nonlinear}) we obtain
    \begin{equation}
        W(t) \leq \beta_1(W(0),t) = \beta_1\left(V(y(0))-C,t\right),
    \end{equation}
    which implies 
    \begin{align}
        V(y(t)) &= W(t) + C \leq \beta_1(V(x(0))-C,t) + C \\
        \overset{\eqref{eq:thm_ISS_0}}&{\leq} \beta_1(\bar\alpha(\lVert x(0) \rVert),t)+C .
    \end{align}
    Again using \eqref{eq:thm_ISS_0} we obtain
    \begin{align}
        \lVert y(t) \rVert &\leq \ubar\alpha^{-1}(V(y(t))) \\
        &\leq \ubar\alpha^{-1}(\beta_1(\bar\alpha(\lVert y(0) \rVert),t)+C) \\
        &\leq  \ubar\alpha^{-1}(2\beta_1(\bar\alpha(\lVert y(0) \rVert),t)) + \ubar\alpha^{-1}(2C) \\
        &\leq \beta(\lVert y(0)\rVert,t) + \eta_1(D_1) + \eta_2(D_2),
    \end{align}
    where $(r,t) \mapsto \beta(r,t) \coloneqq \ubar\alpha^{-1}(2\beta_1(\bar\alpha(r),t))$ is of class $\mathcal{KL}$, and $r \mapsto \eta_1(r) \coloneqq \ubar\alpha^{-1}(4\alpha_1^{-1}(2\gamma_1(r)))$, $r \mapsto \eta_2(r) \coloneqq \ubar\alpha^{-1}(4\alpha_1^{-1}(2\gamma_2(r)))$ are of class $\mathcal{K}_\infty$. The last two inequalities follow from $\mu(a+b) \leq \mu(2\max(a,b))\leq\mu(2a)+\mu(2b)$ for any $a,b\geq 0$ and strictly increasing $\mu$.
\end{proof}


\bibliographystyle{IEEEtran}
\bibliography{IEEEabrv,lib}

\begin{thebibliography}{10}
\providecommand{\url}[1]{#1}
\csname url@rmstyle\endcsname
\providecommand{\newblock}{\relax}
\providecommand{\bibinfo}[2]{#2}
\providecommand\BIBentrySTDinterwordspacing{\spaceskip=0pt\relax}
\providecommand\BIBentryALTinterwordstretchfactor{4}
\providecommand\BIBentryALTinterwordspacing{\spaceskip=\fontdimen2\font plus
\BIBentryALTinterwordstretchfactor\fontdimen3\font minus \fontdimen4\font\relax}
\providecommand\BIBforeignlanguage[2]{{%
\expandafter\ifx\csname l@#1\endcsname\relax
\typeout{** WARNING: IEEEtran.bst: No hyphenation pattern has been}%
\typeout{** loaded for the language `#1'. Using the pattern for}%
\typeout{** the default language instead.}%
\else
\language=\csname l@#1\endcsname
\fi
#2}}

\bibitem{mead_neuromorphic_1990}
C.~Mead, ``Neuromorphic electronic systems,'' \emph{Proceedings of the IEEE}, vol.~78, no.~10, pp. 1629--1636, 1990.

\bibitem{mead2012analog}
C.~Mead and M.~Ismail, \emph{Analog VLSI implementation of neural systems}.\hskip 1em plus 0.5em minus 0.4em\relax Springer Science \& Business Media, 2012, vol.~80.

\bibitem{deweerth1991simple}
S.~DeWeerth, L.~Nielsen, C.~Mead, and K.~Astr{\"o}m, ``A simple neuron servo,'' \emph{IEEE Transactions on Neural Networks}, vol.~2, no.~2, pp. 248--251, 1991.

\bibitem{aaarzen1999simple}
K.-E. {\AA}arz{\'e}n, ``A simple event-based {PID} controller,'' \emph{IFAC Proceedings Volumes}, vol.~32, no.~2, pp. 8687--8692, 1999.

\bibitem{astrom2002comparison}
K.~Astr{\"o}m and B.~Bernhardsson, ``Comparison of {Riemann} and {Lebesgue} sampling for first order stochastic systems,'' in \emph{Proceedings of the 41st IEEE Conference on Decision and Control, 2002.}, vol.~2, 2002, pp. 2011--2016 vol.2.

\bibitem{heemels2012anintroduction}
W.~Heemels, K.~Johansson, and P.~Tabuada, ``An introduction to event-triggered and self-triggered control,'' in \emph{2012 IEEE 51st IEEE Conference on Decision and Control (CDC)}, 2012, pp. 3270--3285.

\bibitem{antunes2016consistent}
D.~J. Antunes and B.~A. Khashooei, ``Consistent event-triggered methods for linear quadratic control,'' in \emph{2016 IEEE 55th Conference on Decision and Control (CDC)}, 2016, pp. 1358--1363.

\bibitem{mahmoud2014networked}
M.~S. Mahmoud and M.~Sabih, ``Networked event-triggered control: an introduction and research trends,'' \emph{International Journal of General Systems}, vol.~43, no.~8, pp. 810--827, 2014.

\bibitem{rieke1999spikes}
F.~Rieke, D.~Warland, R.~d.~R. Van~Steveninck, and W.~Bialek, \emph{Spikes: exploring the neural code}.\hskip 1em plus 0.5em minus 0.4em\relax MIT press, 1999.

\bibitem{ahissar2016perception}
E.~Ahissar and E.~Assa, ``Perception as a closed-loop convergence process,'' \emph{eLife}, vol.~5, p. e12830, 2016.

\bibitem{cisek1999beyond}
P.~Cisek, ``Beyond the computer metaphor: Behaviour as interaction,'' \emph{Journal of Consciousness Studies}, vol.~6, no. 11-12, pp. 125--142, 1999.

\bibitem{moore2024theneuron}
J.~J. Moore, A.~Genkin, M.~Tournoy, J.~L. Pughe-Sanford, R.~R. de~Ruyter~van Steveninck, and D.~B. Chklovskii, ``The neuron as a direct data-driven controller,'' \emph{Proceedings of the National Academy of Sciences}, vol. 121, no.~27, p. e2311893121, 2024.

\bibitem{tabuada2007event}
P.~Tabuada, ``Event-triggered real-time scheduling of stabilizing control tasks,'' \emph{IEEE Transactions on Automatic Control}, vol.~52, no.~9, pp. 1680--1685, 2007.

\bibitem{aranda2020bibliometric}
E.~Aranda-Escolástico, M.~Guinaldo, R.~Heradio, J.~Chacon, H.~Vargas, J.~Sánchez, and S.~Dormido, ``Event-based control: A bibliometric analysis of twenty years of research,'' \emph{IEEE Access}, vol.~8, pp. 47\,188--47\,208, 2020.

\bibitem{goebel2012hybrid}
R.~Goebel, R.~G. Sanfelice, and A.~R. Teel, \emph{Hybrid Dynamical Systems: Modeling, Stability, and Robustness}.\hskip 1em plus 0.5em minus 0.4em\relax Princeton University Press, 2012.

\bibitem{haddad_impulsive_2014}
W.~M. Haddad, V.~Chellaboina, and S.~G. Nersesov, Eds., \emph{\BIBforeignlanguage{eng}{Impulsive and {Hybrid} {Dynamical} {Systems}: {Stability}, {Dissipativity}, and {Control}}}.\hskip 1em plus 0.5em minus 0.4em\relax Princeton University Press, 2014.

\bibitem{donkers2012output}
M.~C.~F. Donkers and W.~P. M.~H. Heemels, ``Output-based event-triggered control with guaranteed {${\cal L}_{\infty}$}-gain and improved and decentralized event-triggering,'' \emph{IEEE Transactions on Automatic Control}, vol.~57, no.~6, pp. 1362--1376, 2012.

\bibitem{heemels2013periodic}
W.~P. M.~H. Heemels, M.~C.~F. Donkers, and A.~R. Teel, ``Periodic event-triggered control for linear systems,'' \emph{IEEE Transactions on Automatic Control}, vol.~58, no.~4, pp. 847--861, 2013.

\bibitem{postoyan2011lyapunov}
R.~Postoyan, A.~Anta, D.~Nešić, and P.~Tabuada, ``A unifying {Lyapunov}-based framework for the event-triggered control of nonlinear systems,'' in \emph{2011 50th IEEE Conference on Decision and Control and European Control Conference}, 2011, pp. 2559--2564.

\bibitem{heemels2008analysis}
W.~P. M.~H. Heemels, J.~H. Sandee, and P.~P. J. V.~D. Bosch, ``Analysis of event-driven controllers for linear systems,'' \emph{International Journal of Control}, vol.~81, no.~4, pp. 571--590, 2008.

\bibitem{yang2002impulsive}
T.~Yang, \emph{Impulsive Control Theory}.\hskip 1em plus 0.5em minus 0.4em\relax Springer Berlin Heidelberg, 2001.

\bibitem{meng2012optimal}
X.~Meng and T.~Chen, ``Optimal sampling and performance comparison of periodic and event based impulse control,'' \emph{IEEE Transactions on Automatic Control}, vol.~57, no.~12, pp. 3252--3259, 2012.

\bibitem{petri2024analysis}
E.~Petri, K.~J. Scheres, E.~Steur, and W.~Heemels, ``Analysis of a simple neuromorphic controller for linear systems: A hybrid systems perspective,'' in \emph{2024 IEEE 63rd Conference on Decision and Control (CDC)}.\hskip 1em plus 0.5em minus 0.4em\relax IEEE, 2024, pp. 8578--8583.

\bibitem{chai2017analysis}
J.~Chai, P.~Casau, and R.~G. Sanfelice, ``Analysis and design of event-triggered control algorithms using hybrid systems tools,'' in \emph{2017 IEEE 56th Annual Conference on Decision and Control (CDC)}, 2017, pp. 6057--6062.

\bibitem{li2020lyapunov}
X.~Li, D.~Peng, and J.~Cao, ``Lyapunov stability for impulsive systems via event-triggered impulsive control,'' \emph{IEEE Transactions on Automatic Control}, vol.~65, no.~11, pp. 4908--4913, 2020.

\bibitem{agliati2025spiking}
P.~Agliati, A.~Urbano, P.~Lanillos, N.~Ahmad, M.~van Gerven, and S.~Keemink, ``Spiking neurons as predictive controllers of linear systems,'' \emph{arXiv preprint arXiv:2507.16495}, 2025.

\bibitem{slijkhuis2023closed}
F.~S. Slijkhuis, S.~W. Keemink, and P.~Lanillos, ``Closed-form control with spike coding networks,'' \emph{IEEE Transactions on Cognitive and Developmental Systems}, 2023.

\bibitem{huang2018dynamical}
F.~Huang and S.~Ching, ``Dynamical spiking networks for distributed control of nonlinear systems,'' in \emph{2018 Annual American Control Conference (ACC)}, 2018, pp. 1190--1195.

\bibitem{thalmeier2016learning}
D.~Thalmeier, M.~Uhlmann, H.~J. Kappen, and R.-M. Memmesheimer, ``Learning universal computations with spikes,'' \emph{PLoS Computational Biology}, vol.~12, no.~6, p. e1004895, 2016.

\bibitem{petri2025spiking}
E.~Petri, K.~J.~A. Scheres, E.~Steur, and W.~P. M.~H. Heemels, ``Spiking control for the stabilization of linear time-invariant systems: An emulation-based design approach,'' in \emph{2025 IEEE 64th Conference on Decision and Control (CDC)}, 2025, pp. 7--12.

\bibitem{petri2025emulation}
\BIBentryALTinterwordspacing
E.~Petri, K.~J.~A. Scheres, E.~Steur, W.~P.~M. H., and Heemels, ``Emulation-based neuromorphic control for the stabilization of {LTI} systems,'' 2025. [Online]. Available: \url{https://arxiv.org/abs/2511.11875}
\BIBentrySTDinterwordspacing

\bibitem{dimarogonas2012distributed}
D.~V. Dimarogonas, E.~Frazzoli, and K.~H. Johansson, ``Distributed event-triggered control for multi-agent systems,'' \emph{IEEE Transactions on Automatic Control}, vol.~57, no.~5, pp. 1291--1297, 2012.

\bibitem{wang2011event}
X.~Wang and M.~D. Lemmon, ``Event-triggering in distributed networked control systems,'' \emph{IEEE Transactions on Automatic Control}, vol.~56, no.~3, pp. 586--601, 2011.

\bibitem{kofman2006level}
E.~Kofman and J.~H. Braslavsky, ``Level crossing sampling in feedback stabilization under data-rate constraints,'' in \emph{Proceedings of the 45th IEEE Conference on Decision and Control}, 2006, pp. 4423--4428.

\bibitem{gerstner2002spiking}
W.~Gerstner and W.~M. Kistler, \emph{Spiking neuron models: Single neurons, populations, plasticity}.\hskip 1em plus 0.5em minus 0.4em\relax Cambridge university press, 2002.

\bibitem{khalil2002nonlinear}
H.~K. Khalil and J.~W. Grizzle, \emph{Nonlinear systems}.\hskip 1em plus 0.5em minus 0.4em\relax Prentice hall Upper Saddle River, NJ, 2002, vol.~3.

\bibitem{Sontag2008Input}
E.~D. Sontag, \emph{Input to State Stability: Basic Concepts and Results}.\hskip 1em plus 0.5em minus 0.4em\relax Berlin, Heidelberg: Springer, 2008, pp. 163--220.

\end{thebibliography}

\end{document}